\documentclass{llncs}
\pdfoutput=1

\usepackage[utf8]{inputenc}
\usepackage[english]{babel}
\usepackage[T1]{fontenc}
\usepackage{amsmath}
\usepackage{amsfonts}
\usepackage{amssymb}
\usepackage{hyperref}
\usepackage{algpseudocode}
\usepackage{tikz}
\usetikzlibrary{positioning, backgrounds, arrows}
\usetikzlibrary{fit,shapes.geometric,shapes.arrows,calc,decorations.text}
\usepackage{subfig}
\graphicspath{{figures/}}
\usepackage{wrapfig}

\usepackage{cutwin}
\usepackage{caption}
\opencutright

\usepackage[ampersand]{easylist}

\usepackage{ms}

\title{Is my attack tree correct?}
\subtitle{Extended version\thanks{This is an extended 
version of the work published at ESORICS 2017~\cite{esorics17}.}}

\author{Maxime Audinot\inst{1,2} \and Sophie Pinchinat\inst{1,2} 
\and Barbara Kordy\inst{1,3}}
\institute{IRISA, Rennes, France \and
University Rennes~1, Rennes, France \and
INSA Rennes, France}

\begin{document}

\maketitle

\begin{abstract}
Attack trees are a popular way to represent and evaluate potential security 
threats on  systems or infrastructures. 
The goal of this work is to provide a framework allowing to express and check  
whether an attack tree is consistent with the analyzed system. 
We model real systems using transition systems and introduce
attack trees with formally specified node labels. 
We formulate the correctness properties of an attack tree with respect to 
a system and study the complexity of the corresponding decision problems. 
The proposed framework can be used in practice to assist 
security experts in manual creation of attack trees and enhance development 
of tools for automated generation of attack trees. 
\end{abstract}

\section{Introduction}
\label{sec:intro}
An attack tree is a graphical model allowing a security expert to illustrate 
and analyze potential security threats. 
Thanks to their intuitiveness, attack trees gained a 
lot of popularity in the industrial sector~\cite{survey}, 
and organizations such as 
NATO~\cite{OTAN} and 
OWASP~\cite{OWASP} recommend their use in threat 
assessment processes. 
The root of an attack tree represents an attack objective, \ie an 
attacker's goal, and the rest of the tree decomposes this goal into sub-goals 
that the attacker may need to reach in order to perform his 
attack~\cite{Schn}. 
In this paper, we develop a formal framework to evaluate
\emph{how well an attack tree describes the attacker's goal 
with respect to 
the system that is being analyzed}. 
This work has been motivated by the 
two following practical problems.

First, in the industrial context, attack trees 
are usually created manually by security experts who may not have an 
exhaustive knowledge about all the facets (technical, social, physical) of 
the analyzed system. This process is often supported by 
the use of libraries containing generic models for 
standard security threats. Although using libraries provides a good 
starting point, the resulting attack tree may not 
always be fully consistent with the system that is being analyzed.
This problem might be reinforced by the fact that the node names 
in attack trees are often very short, and may 
thus lack precision or be inaccurate and misleading. 
If the tree is incomplete or imprecise, the results of its analysis 
(\eg estimation of the attack's cost or its probability) 
might be inaccurate. If the tree contains branches that are irrelevant for 
the considered system, the time of its analysis might be longer than 
necessary. This implies that a manually created tree needs to be validated 
against a system to be analyzed before it can be used as a formal model 
on which the security of the system will be evaluated. 

Second, to limit the burden of their manual creation, 
several academic proposals for 
automated generation of attack trees have recently been 
made~\cite{vigo_generation,pinchinat2015atsyra,trespass_generation}. 
In particular, 
we are currently developing the ATSyRA tool 
for assisted generation of attack trees from system 
models~\cite{pinchinat2015atsyra}.
Our experience shows that, due to 
the complexity and scalability issues, a fully automated generation 
is impossible. Some generation steps must thus be supported by humans. 
Such a semi-automated approach gives the expert a possibility of manually 
decomposing a goal, in such a way that an automated generation 
of the subtrees can be performed. 
This work provides 
formal foundations for the next version of our tool
which will assist the expert in producing 
trees that, by design, are correct with respect to the underlying system.

\paragraph{Contribution.}
To address the problems identified above, we introduce a mathematical 
framework 
allowing us to formalize the notion of attack trees and to define as well as 
verify their practically-relevant correctness properties with respect to a 
given system. 
We model real-life systems using finite 
transition systems. 
The attack tree nodes are \emph{labeled with
formally specified goals formulated in terms of preconditions and 
postconditions} over the possible states of the transition system. 
Formalizing the labels of the attack tree nodes allows us to 
overcome the problem of imprecise or 
misleading text-based node names 
and makes formal treatment of attack trees 
possible. 
We define the notion of 
\emph{Admissibility} of an attack tree with respect to a given system 
and introduce the correctness properties for attack trees, called  
\emph{Meet}, \emph{Under-Match}, 
\emph{Over-Match}, and \emph{Match}. These properties 
express the precision with which a given goal is refined into 
sub-goals with respect to a given system. 
We then \emph{establish the complexity of
verifying the correctness properties} to apprehend the nature of 
potential algorithmic solutions to be implemented.

\paragraph{Related work.}
In order to use any modeling framework in practice, formal foundations are 
necessary. Previous research on formalization of attack trees 
focused mainly on mathematical semantics for attack tree-based 
models~\cite{mauw2005foundationsatrees,JurgensonW09,kordy2012adtrees,jhawar2015seqatree,FundInf17},
and various algorithms for their quantitative 
analysis~\cite{act,bayesian,post15}.  
However, all these formalizations rely on an 
action-based approach, where 
the attacker's goals 
represented by the labels of the attack tree nodes 
are expressed using actions that the attacker needs to perform 
to achieve his/her objective. 
In this work, we pioneer a state-based approach to attack trees, 
where the attacker's goals relate to 
the states of the modeled system. The advantage of 
such a state-based approach is that it may benefit
from verification and model checking techniques, in a natural way, 
as this has already been done in the case of attack 
graphs~\cite{SheynerHJLW02,PhillipsS98}. 
In our framework, the label of each node of an attack tree 
is formulated in terms of preconditions and postconditions over 
the states of the modeled system: intuitively speaking, 
the goal of the attacker is to start from any state in the system that 
satisfies the preconditions and reach a state where the postconditions are met. 
The idea of formalizing the labels of attack tree nodes  in terms of 
preconditions and postconditions has already been explored in~\cite{Wolter}. 
However there, the postcondition (\ie consequence) of an action is 
represented by a parent node and its children model the preconditions 
and the action itself. 

Model checking of attack trees, especially using tools such as PRISM or UPPAAL, has
already been successfully employed, in particular to
support their quantitative 
analysis, as in~\cite{Gadyatskaya2016modelcheckingquantitative,Kumar2015uppaal,aslanyan2017exactcosts}. Such
techniques provide an effective way of handling a multi-parameter
evaluation of attack scenarios, \eg identifying the resources
needed for a successful attack or checking whether 
there exists an attack whose cost is lower than a given value and whose
probability of success is greater than a certain threshold. However, 
these approaches either do not consider any particular system beforehand, 
or they rely on a model of the system that features explicit quantitative
aspects.
The link between the analyzed system and the corresponding attack tree 
is made explicit in works dealing with automated generation of 
attack trees from system 
models~\cite{trespass_generation,pinchinat2015atsyra}. 
The systems considered in~\cite{trespass_generation} capture 
locations, assets, processes, policies, and actors. The goal of the attacker 
is to reach a given location or obtain an asset, and 
the attack tree generation algorithm relies on invalidation of 
policies that forbid him to do so. 
In the case of~\cite{pinchinat2015atsyra}, 
the ATSyRA tool is used to  
effectively generate a transition system for a real-life system: 
starting from a domain-specific language describing the original system, 
ATSyRA compiles this description into a 
symbolic transition system specified in the guarded action language 
GAL~\cite{gal}. 
ATSyRA can already handle the physical layer of a system 
(locations and connections/accesses between them)
and we are currently working on extending it with the digital layer. 
Since our experience shows that generating a transition system 
from a description in a domain-specific language is possible and efficient, 
in this paper we suppose that the transition system for 
a real system has been previously created and is available.

Finally, to the best of our knowledge, the 
problem of defining and verifying the correctness 
of an attack tree with respect to the analyzed system has only been considered 
in~\cite{audinot2016soundnessatree} which has been the starting point for the 
work presented in this paper.

\section{Motivating example}
\label{sec:cs}
Before presenting our framework, we first introduce 
a motivating example on which we will illustrate the notions and concepts 
employed in this paper.

The system modeled in our running example is a building containing a safe 
holding a confidential document. The goal of the attacker is to reach 
the safe without being detected. 
We purposely keep this example small and intuitive to ease the 
understanding of the proposed framework.
The floor plan of the building is depicted in 
Fig.~\ref{subfig:csbuilda}. It contains two rooms, denoted by 
Room1 and Room2, 
two doors -- Door1 allowing to move from outside of the building to Room1
and Door2 connecting Room1 and Room2 -- as well as one window in Room2. 
Both doors are initially locked and it is left unspecified whether the window 
is open or not. Such unspecified information expresses that the analyst
cannot predict whether the window will be open or closed in the case of 
a potential attack or that he has a limited knowledge about the system. In 
both cases, this lack of information needs to 
be taken into account during the analysis process. The two doors can be 
unlocked by means of Key1 and Key2, respectively. 
We assume that a camera that monitors Door2
is located in Room1. The camera is initially on 
but it can be switched off manually. The safe is in Room2.

\begin{wrapfigure}{r}{0.55\textwidth}
  \centering
  \subfloat[Floor plan]{
    \label{subfig:csbuilda} 
        \def\svgwidth{80pt}
    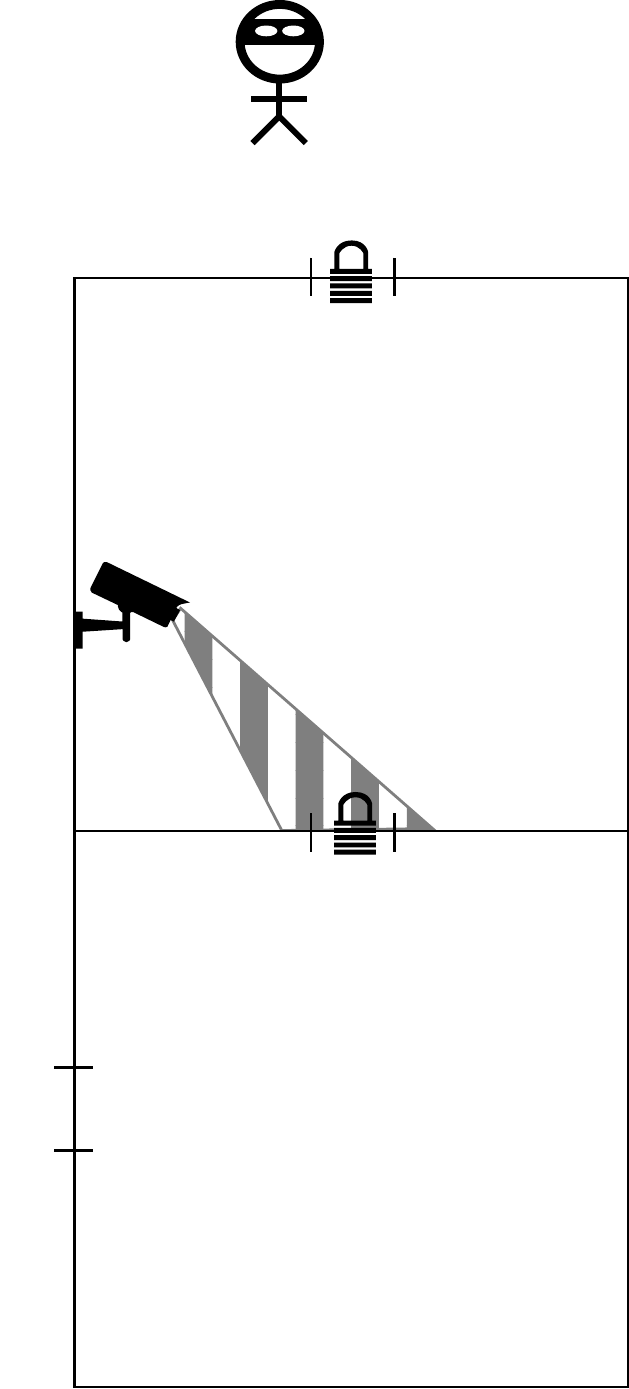
  }
  \subfloat[Attack scenarios]{
    \label{subfig:csbuildb}
        \def\svgwidth{85pt}
    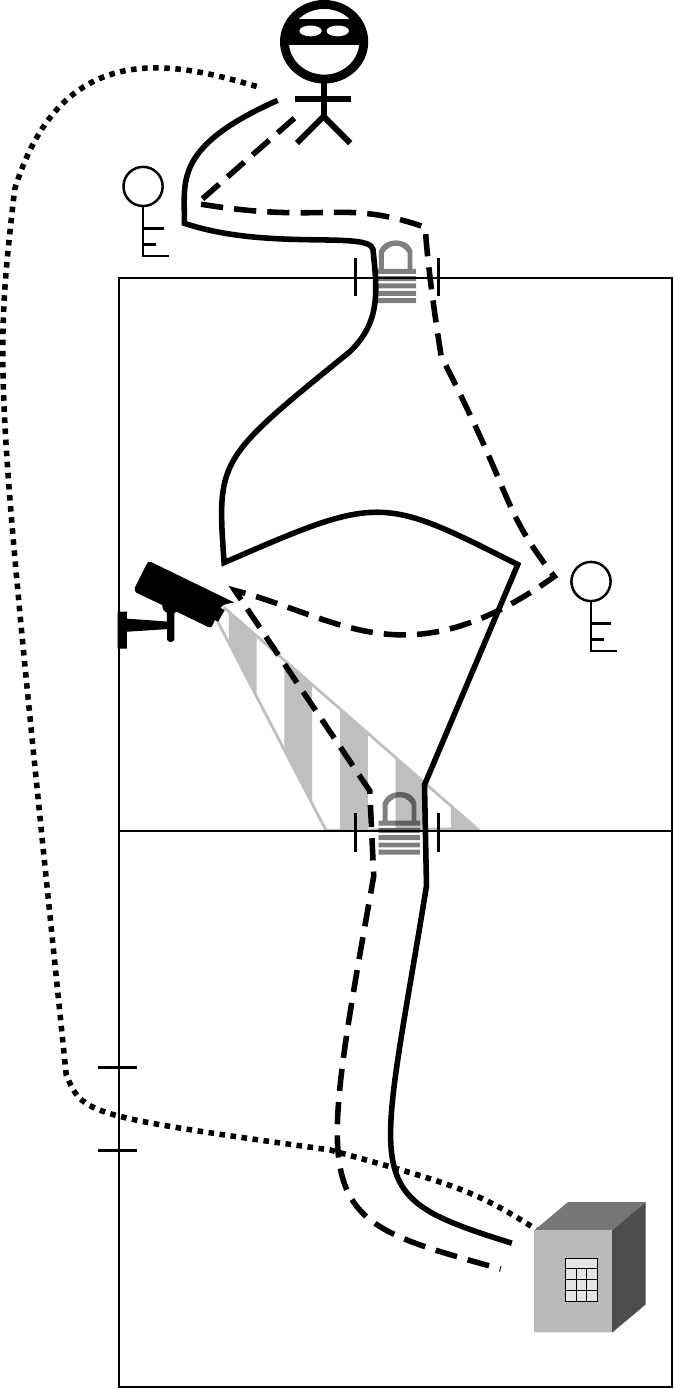
  }
  \caption{Running example building}
  \label{fig:csbuild}
	  \vspace{-20pt}
  \vspace{1pt}
\end{wrapfigure}
The attacker is located outside of the building 
and his goal is to \emph{reach the safe without being detected by the camera}. 
In Fig.~\ref{subfig:csbuildb}, we have depicted three scenarios 
(that we will call paths) allowing the attacker to reach his goal. 
In the first scenario (depicted using dotted line), the attacker goes straight through the window, if 
it is open. In the remaining two scenarios, the attacker gathers the 
necessary keys and goes through the two doors, 
switching off the camera on his way. 
These two scenarios differ only in the order in which the 
concurrent actions are sequentially performed. Since collecting Key2 and 
switching off the camera are independent actions, the attacker can first 
collect Key2 and then switch the camera off (dashed line), 
or switch the camera off before collecting Key2 (solid line).

The \emph{system} in our example consists of the building 
and the attacker. It is modeled using state variables whose values
determine possible configurations of the system. 
\begin{itemize}
\item $\posi$ -- variable 
  describing the attacker's position, ranging over $\{\outside,$ 
	$\roomOne,\roomTwo\}$; 
\item $\wind$ -- Boolean variable describing whether 
the window is open (\true) or not (\false); 
\item $\lockOne$ and $\lockTwo$ -- Boolean variables 
to describe whether the respective doors are locked or not; 
\item $\keyOne$ and $\keyTwo$ -- Boolean variables to describe 
whether the attacker
  possesses the respective key; 
\item $\cameraOn$ -- Boolean variable describing if the camera is on; 
\item $\detected$ -- Boolean variable to describe if  
the camera detected the attacker, \ie whether the attacker has crossed the area monitored by the camera while it was on. 
\end{itemize}

Given a set of state variables, we express possible configurations of a 
system using propositions. \emph{Propositions} are either 
equalities of the form \texttt{state\_variable$=$value} or 
Boolean combinations of such equalities. 
Intuitively, a proposition expresses a constraint on the possible configurations.
A configuration in which all the variables 
are left unspecified is called the \emph{empty configuration}. 
We denote it by $\emptyconf$. 

In order to analyze the security of a system, security experts often use the 
model of attack trees. 
An \emph{attack tree} is a tree in which each node
represents an attacker objective, and the children of a node represent a
decomposition of this objective into sub-objectives.  
In this work, we consider attack trees with three types of nodes: 
\begin{itemize}
\item \OR nodes representing alternative choices -- to achieve 
the goal of the node, the attacker needs to achieve the goal of at least one child; 
\item  \AND nodes representing conjunctive decomposition --
to achieve 
the goal of the node, the attacker needs to achieve all of the 
goals represented by its children (the children of an \AND node are 
connected with an arc); 
\item \SAND nodes representing sequential decomposition -- 
to achieve 
the goal of the node, the attacker needs to achieve  all of the 
goals represented by its children in the given order
(the children of a \SAND node are 
connected with an arrow).   
\end{itemize}

\begin{wrapfigure}{r}{0.4\textwidth}
\vspace*{-5pt}
  \begin{center}
\includegraphics[width=0.35\textwidth]{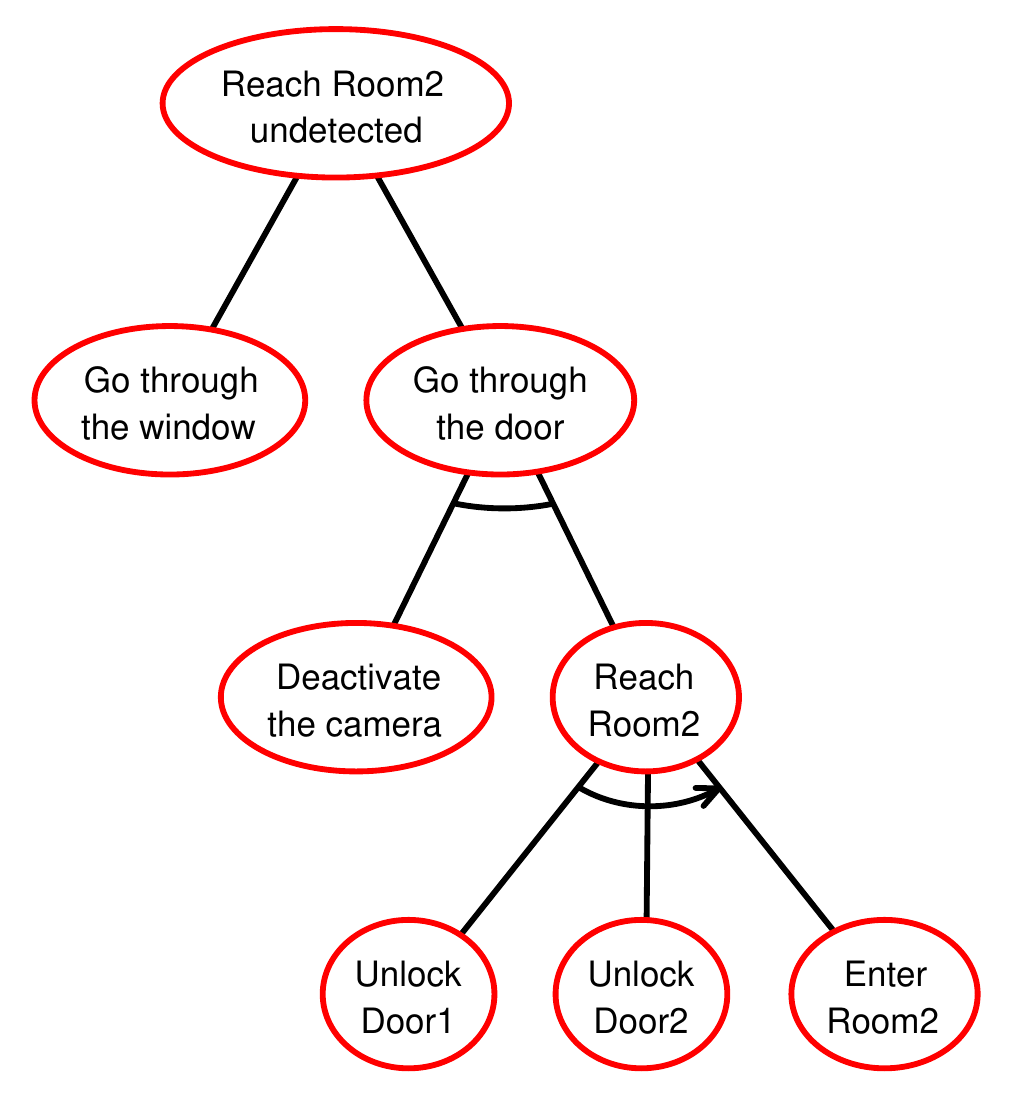}
\caption{Attack tree with informal, text-based node names}
\label{fig:textual_tree}
  \end{center}
  \vspace{-35pt}
  \vspace{1pt}
\end{wrapfigure}
The attack tree given in 
Fig.~\ref{fig:textual_tree} illustrates that in order to enter  Room2 
undetected (root node of type \OR), 
the attacker can either enter through the window or through the doors. 
In order to use the second alternative (node of type \AND), 
he needs to make sure that the camera is deactivated and 
that he reaches Room2. 
To achieve the last objective (node of type \SAND), 
he first needs to unlock 
Room1, then unlock Room2, and finally enter to Room2. 

One of the most problematic aspects of attack trees are 
the informal, text-based names of their nodes. 
These names are
often very short and thus do not express all the information that the
tree author had in mind while creating the tree.  In particular, the
textual names relate to the objective that the attacker should reach, 
however, they usually do
not capture the information about the initial situation from which he starts.

To overcome the weakness of text-based node names, we propose to 
formalize the attacker's goal using two configurations: 
the \emph{initial configuration}, usually denoted by $\atinit$, 
is the configuration before the attack starts, i.e., represents 
preconditions; and the 
\emph{final configuration}, usually denoted by $\atfinal$, 
represents postconditions, \ie the state 
to be reached to succeed in the attack. 
The goal with initial configuration $\atinit$ and final configuration
$\atfinal$ is written $\ag{\atinit}{\atfinal}$.

In our running example, the initial configuration is
$\atinit \egdef (\posi = \outside)\et (\keyOne = \false)\et (\keyTwo = \false)
  \et (\lockOne = \true)\et (\lockTwo = \true)\et (\cameraOn = \true)$.
It describes that the attacker is originally outside of the building, he does 
not have any of the keys, the two doors are locked, and the 
camera is on. 
The final configuration is $\atfinal \egdef (\posi = \roomTwo)\et (\detected =
  \false)$, \ie the attacker reached Room2 without 
being detected.

Fig.~\ref{fig:formal_atree} illustrates how such formally specified goals 
are used to label the nodes of attack trees.
The goal $\ag{\atinit}{\atfinal}$ introduced above is the label of 
 the root node of 
the tree. It is then refined into sub-goals  $\ag{\atinit_i}{\atfinal_i}$, 
where $i$ reflects the position of the node in the tree.

\noindent
\textbf{Sub-goal $\ag{\atinit_1}{\atfinal_1}$}:
The attacker, who wants to reach the safe 
in Room2 without being detected, is 
located outside of the building and the 
window is initially open. We let 
  $\atinit_1 \egdef (\posi = \outside)\et (\keyOne = \false)\et
                     (\keyTwo = \false)\et (\lockOne = \true)
  \et (\lockTwo = \true)\et (\cameraOn = \true)\et (\wind = \true)$ and $\atfinal_1 \egdef \atfinal$. 	\\
	
\noindent
\textbf{Sub-goal $\ag{\atinit_2}{\atfinal_2}$}: This sub-goal
is similar to the previous one, but the window is
originally closed. 
We let $\atinit_2 \egdef (\posi = \outside)\et (\keyOne = \false)\et (\keyTwo = \false) \et (\lockOne = \true) \et (\lockTwo = \true)\et (\cameraOn = \true) \et (\wind = \false)$ and $\atfinal_2 \egdef \atfinal$.\\

\noindent
\textbf{Sub-goal $\ag{\atinit_{21}}{\atfinal_{21}}$}: The attacker, who might be in any initial
configuration, wants to deactivate the camera. We then let $\atinit_{21}
\egdef \emptyconf$ and $\atfinal_{21} \egdef (\cameraOn =\false)$.\\

\begin{wrapfigure}{r}{0.45\textwidth}
  \vspace{-15pt}
  \begin{center}
    \includegraphics[width=0.45\textwidth]{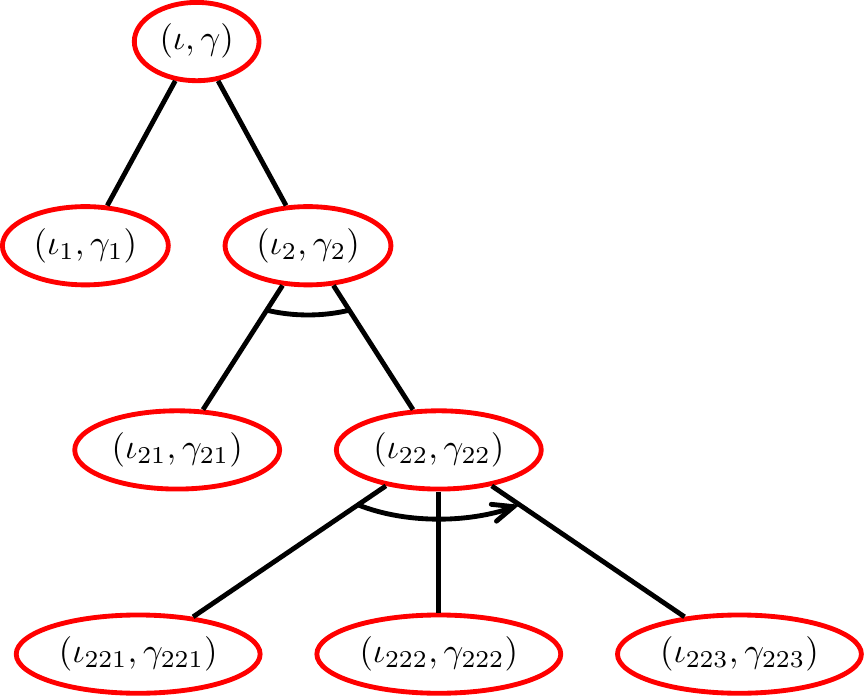}
       \caption{Attack tree with formal labels}
    \label{fig:formal_atree}
  \end{center}
\vspace{-20pt}
\end{wrapfigure}
\noindent
\textbf{Sub-goal $\ag{\atinit_{22}}{\atfinal_{22}}$}: Similar to sub-goal
$\ag{\atinit_2}{\atfinal_2}$, with the difference that we do not care
whether the camera is initially on and we no longer
require that the attacker remains undetected. We let $\atinit_{22}\egdef
(\posi = \outside) \et (\keyOne = \false)\et (\keyTwo = \false) \et
(\lockOne = \true)\et (\lockTwo = \true) \et (\wind = \false)$ and
$\atfinal_{22} \egdef (\posi = \roomTwo)$. \\

 \noindent
 \textbf{Sub-goal $\ag{\atinit_{221}}{\atfinal_{221}}$}:
The initial situation is the same as in the sub-goal
 $\ag{\atinit_{22}}{\atfinal_{22}}$, but we require that the attacker
 unlocks Door1 but not Door2: $\atinit_{221} \egdef \atinit_{22}$ and
 $\atfinal_{221} \egdef (\lockOne = \false) \et (\lockTwo = \true)$. \\

   \noindent
 \textbf{Sub-goal $\ag{\atinit_{222}}{\atfinal_{222}}$}: 
 Now, the objective is to go from a state where 
 Door1 is unlocked and Door2 is locked (like in the configuration $\atfinal_{221}$)
 to a state where both doors are unlocked. We let $\atinit_{222} \egdef \atfinal_{221}$ and $\atfinal_{222} \egdef (\lockOne = \false)\et (\lockTwo =
   \false)$. \\

 \noindent
 \textbf{Sub-goal $\ag{\atinit_{223}}{\atfinal_{223}}$}: 
 Finally, the last sub-goal is for the attacker, starting in a state 
 where both doors are unlocked, to reach  Room2. 
 We let $\atinit_{223} \egdef \atfinal_{222}$ and $\atfinal_{223} \egdef \atfinal_{22}$.

\section{Formal modeling}

We now provide formal notations and 
definitions of transition systems and attack trees 
that we have informally described in Sect.~\ref{sec:cs}.

\subsection{Transition systems}
\label{sec:TS}

We model real-life systems using finite 
transition systems. 
Transition system is a simple, yet powerful
formal tool to represent a dynamic behavior of a system 
by listing all its possible states and transitions between them. 
The finiteness of the state transition system 
is a reasonable and realistic assumption. A formal model can either be 
finite because the real-life underlying system is intrinsically finite, or it 
can have a finite representation obtained by standard abstraction techniques, 
as used in verification, static analysis, and model-checking. 

We fix the set $\AtmSet$ of propositions that 
we use to formalize possible configurations of the real system. 
In the rest of the paper, we suppose that $\AtmSet$ contains propositions
of the form $\atinit, \atfinal$, to denote preconditions ($\atinit$) and postconditions
($\atfinal$) of the  goals. 

\begin{definition}[Transition system]
  \label{def:system}
  \label{def:postpre}~\\
  A \emph{transition system} over $\AtmSet$ is a tuple
  $\system = \systemtuple$, where 
$\StateSet$ is a finite set of \emph{states} (elements of 
	$\StateSet$ are denoted 
    by $s, s_i$ for $i\in\mathbb{N}$), $\Transition \subseteq \StateSet \times \StateSet$ is the
    \emph{transition relation} of the system (which is assumed
    left-total), and $\valuation : \AtmSet \rightarrow 2^{\StateSet}$ is the
    \emph{labeling} function. We say that a state $s$ \emph{is labeled by} $p$
    when $s \in \valuation(p)$.   The \emph{size} of $\system$ is
  $\size{\system} = \size{\StateSet} + \size{\Transition}$.
  \end{definition}
For the rest of this paper, we assume that we are given a transition system $\system$ over $\AtmSet$.
A \emph{path} in $\system$ is a non-empty sequence of states. We use
typical elements $\pathvar, \pathvar',\pathvar_1, \ldots,$ $ \rho, \dots$ to
denote paths. The \emph{size} of a path $\pathvar$, denoted by
$\size{\pathvar}$, is its number of transitions, and $\pathvar(i)$ is 
the element at position $i$ in $\pathvar$, for $0 \leq i \leq \size{\pathvar}$. An empty
path\footnote{Since a path is a non-empty sequence of states, the
  empty path contains exactly one state.} is a path of size $0$.  We
write $\SPathSet$ for the set of all paths in $\system$. For
$\atinit,\atfinal \in \AtmSet$, we shortly say that a path $\pathvar$ ``goes from $\atinit$
to $\atfinal$'' whenever $\pathvar(0) \in \valuation(\atinit)$ and
$\pathvar(\size{\pathvar}) \in \valuation(\atfinal)$. 
The set of \emph{direct successors} of a set of states
  $\StateSet' \subseteq \StateSet$ is
  $\post(\StateSet') = \Set{s \in \StateSet \suchthat \exists 
	s' \in \StateSet' \text{ such that } (s',s) \in
    \Transition}$. The set of \emph{successors} of a set of states
$\StateSet' \subseteq \StateSet$ is
$\reach(\StateSet') = \Set{s \in \StateSet \suchthat \exists \pathvar \text{ with } \pathvar(0) \in \StateSet' \text{ and }
  \pathvar(\size{\pathvar}) = s}$, and the set of \emph{predecessors} of
$\StateSet' \subseteq \StateSet$ is
$\coreach(\StateSet') = \Set{s \in \StateSet \suchthat \exists \pathvar \text{ with } \pathvar(0) = s \text{ and }
  \pathvar(\size{\pathvar}) \in \StateSet'}$.

A factor of a path $\pathvar$ is a subsequence composed 
of consecutive elements of $\pathvar$. 
Formally, a \emph{factor} of a path $\pathvar$ is a path $\pathvar'$, 
such that there exists $0 \leq k \leq \size{\pathvar} - \size{\pathvar'}$, where
$\pathvar(i+k) = \pathvar'(i)$, for $0 \leq i \leq \size{\pathvar'}$. 
An \emph{anchoring} of $\pathvar'$ in $\pathvar$ is an interval
$\interval{k}{l} \subseteq \interval{0}{\size{\pathvar}}$ where
for all $i \in \interval{k}{l}$, $\pathvar'(i-k) = \pathvar(i)$ and 
$l - k = \size{\pathvar'}$. Notice that we may have $\size{\pathvar'}=0$.
We denote by $\pathvar\interval{k}{l}$ 
the factor of $\pathvar$ of anchoring
$\interval{k}{l}$.
In other words, the anchorings of $\pathvar'$ in $\pathvar$ 
are the intervals $\interval{k}{l}$ of positions in $\pathvar$ such that $\pathvar\interval{k}{l}=\pathvar'$.

We now introduce concatenation and parallel 
decomposition of paths -- two notions that will serve us to define the semantics of  
sequential and conjunctive refinements in attack trees, respectively. 
\begin{definition}[Concatenation of paths]
\label{sef:concatenation}
Let $\pathvar_1,\pathvar_2,\ldots,\pathvar_n \in \SPathSet$ be paths, such that
$\pathvar_i(\size{\pathvar_i}) = \pathvar_{i+1}(0)$ for $1 \leq i \leq n-1$.  The \emph{concatenation} of
$\pathvar_1, \pathvar_2,\ldots,\pathvar_n $, denoted by $\pathvar_1 . \pathvar_2. \ldots .\pathvar_n$, is the path
$\pathvar$, where $\pathvar\interval{\sum_{k=1}^{i-1}\size{\pathvar_k}}{\sum_{k=1}^{i-1}\size{\pathvar_k} + \size{\pathvar_{i}}} = \pathvar_{i}$\footnote{We use the convention that $\sum_{k=1}^0\size{\pathvar_k}=0$.}.
We generalize the concatenation to sets of paths by letting
$\PathSet . \PathSet' = \Set{\pathvar \in \SPathSet \suchthat
  \exists i, 0 \leq i\leq \size{\pathvar} \text{ and }
  \pathvar\interval{0}{i} \in \PathSet \text{ and }
  \pathvar\interval{i}{\size{\pathvar}} \in \PathSet'}$.
\end{definition}

\begin{definition}[Parallel decomposition of paths]
  \label{def:decomp}
 A set $\Set{\pathvar_1,\dots,\pathvar_n} \subseteq \SPathSet$ is a
 \emph{parallel decomposition} of $\pathvar \in \SPathSet$ if for every
 $1\leq i \leq n$ the path $\pathvar_i$ is a factor of $\pathvar$ for some 
 anchoring $\interval{k_i}{l_i}$, such that every interval
 $\interval{j}{j+1} \subseteq \interval{0}{\size{\pathvar}}$ is contained
 in $\interval{k_i}{l_i}$ for some $i\in\{1,\dots, n\}$ 
(which trivially holds if $\size{\pathvar}=0$). We then say that the sequence $\pathvar_1,\dots,\pathvar_n$ is a
 \emph{parallel decomposition} of $\pathvar$ for the anchorings $\interval{k_1}{l_1},\ldots,\interval{k_n}{l_n}$.
  \end{definition}

\begin{lemma}
  \label{lem:decomp} 
  Given a path $\pathvar \in \SPathSet$, and a sequence $k_1,l_1,\ldots,k_n,l_n \in
  \interval{0}{\size{\pathvar}}$, deciding whether 
  $\pathvar\interval{k_1}{l_1},\ldots,\pathvar\interval{k_n}{l_n}$ is
  a parallel decomposition of $\pathvar$ for the anchorings $\interval{k_1}{l_1},\ldots,$ $\interval{k_n}{l_n}$ can be done in time $\OO{n\size{\pathvar}}$.
\end{lemma}
\begin{proof}
    Verifying that $\pathvar\interval{k_1}{l_1},\ldots,\pathvar\interval{k_n}{l_n}$ is
  a parallel decomposition of $\pathvar$ for the anchorings $\interval{k_1}{l_1},\ldots,\interval{k_n}{l_n}$ amounts to checking that for 
   every interval $\interval{j}{j+1} \subseteq \interval{0}{\size{\pathvar}}$, there is an $i\in\Setleqn{n}$ such that $\interval{j}{j+1}
   \subseteq \interval{k_i}{l_i}$. This can clearly be done in time $\OO{n\size{\pathvar}}$ by a naive approach. \qed
\end{proof}

An example of a parallel decomposition is illustrated in Fig.~\ref{fig:assos}, 
where $\pi_1=\pi\interval{0}{2}$, $\pi_2=\pi\interval{3}{5}$, and 
$\pi_3=\pi\interval{1}{4}$. 
\begin{figure}[ht]
\centering{
\begin{tikzpicture}[scale=1]
\coordinate[thin] (1) at (0,2)  ; 
\coordinate[thin] (2) at (2,2); 
\coordinate[thin] (3) at (2.5,2) ; 
\coordinate[thin] (4) at (4,2);
\coordinate[thin] (5) at (4.5,2) ; 
\coordinate[thin] (6) at (6,2);
 \draw (1)node[below] {$s_0$} node[above] {$\atinit$} node {\textbullet}--(2)
node[below] {$s_1$} node {\textbullet}--(3)
node[below] {$s_2$} node {\textbullet}-- node[above]{$\pi$} (4)
node[below] {$s_3$} node {\textbullet}--(5)
node[below] {$s_4$} node {\textbullet}--(6) 
node[below] {$s_5$} node[above] {$\atfinal$} node {\textbullet};
\coordinate[thin] (11) at (0,1) ; 
\coordinate[thin] (13) at (2.5,1) ; 
\draw (11)node[below] {$s_0$} 
node[above] {$\atinit_1$} node {\textbullet}-- node[above]{$\pi_1$} (13)
node[below] {$s_2$} node[above] {$\atfinal_1$} node {\textbullet};
\coordinate[thin] (41) at (4,1) ; 
\coordinate[thin] (46) at (6,1) ; 
\draw (41)node[below] {$s_3$} 
node[above] {$\atinit_2$}
node {\textbullet}-- node[above]{$\pi_2$}(46)
node[below] {$s_5$} node[above] {$\atfinal_2$} 
node {\textbullet};
\coordinate[thin] (23) at (2,0) ; 
\coordinate[thin] (53) at (4.5,0) ; 
\draw (23)node[below] {$s_1$} 
node[above] {$\atinit_3$}
node {\textbullet}-- node[above]{$\pi_3$} (53)
node[below] {$s_4$} 
node[above] {$\atfinal_3$}
node {\textbullet};
\end{tikzpicture}
}
\caption{Parallel decomposition of $\pi$ into $\{\pi_1, \pi_2, \pi_3\}$.}
\label{fig:assos}
\end{figure}
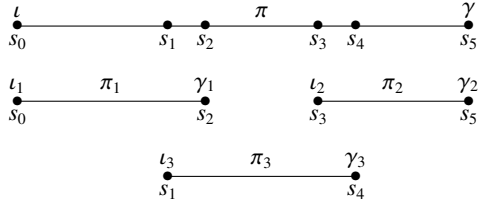

A \emph{cycle} in a path $\pathvar\in\SPathSet$ is a factor $\pathvar'$ of
$\pathvar$ such that $\fst{\pathvar'} = \lst{\pathvar'}$. An \emph{elementary
  path} is a path with no cycle.  Remark that an elementary path
$\pathvar$ does not contain any state more than once, so $\size{\pathvar}
\leq \size{\system}$.  Removing  a cycle $\pathvar'$
of anchoring $\interval{k}{l}$ from a path $\pathvar$ yields the path
$\pathvar[0,k].\pathvar[l,\size{\pathvar}]$.  Removing all the cycles from
$\pathvar$ consists in iteratively removing cycles until the
resulting path is elementary. Note that the resulting path may depend
on the order in which the cycles are removed.

We illustrate the notions defined in this 
section on our running example. 
\begin{example}
\label{ex:system}
We use the state variables introduced in Sect.~\ref{sec:cs} to 
describe the states of a part of our building system. 
By $\etat_0$ we denote the state where 
$\posi = \outside$ (the attacker is outside); $\wind = \false$  (the window is closed); $\lockOne= \lockTwo = \true$ (both doors are locked); $\keyOne = \keyTwo = \false$ (the attacker does not have any key); $\cameraOn = \true$ (the camera is on); $\detected = \false$ (the attacker has not been detected).
Furthermore, we consider seven additional states $\etat_i$, such that, 
for every $1\leq i\leq 7$, the specification of $\etat_i$ is the same 
as the specification of 
$\etat_{i-1}$, except one variable: state $\etat_1$ is as $\etat_0$ but $\keyOne=\true$  (the attacker has Key1); state $\etat_2$ is as $\etat_1$ but $\lockOne = \false$ (Door1 is unlocked); state $\etat_3$ is as $\etat_2$ but $\posi=\roomOne$ (the attacker is in Room1); $\etat_4$ is as $\etat_3$ but $\cameraOn = \false$ (the camera is off); $\etat_5$ is as $\etat_4$ but $\keyTwo=\true$  (the attacker has Key2); state $\etat_6$ is  as $\etat_5$ but $\lockTwo = \false$ (Door2 is unlocked); state $\etat_7$ is as $\etat_6$ but $\posi=\roomTwo$ (the attacker is in Room2).

To model the dynamic behavior of the system, we set 
$(\etat_{i-1},\etat_{i})\in \Transition$,  for all $1\leq i\leq 7$. 
Given 
$p=(\posi = \outside)\et (\lockOne=\true)$ and 
$p'= (\posi = \roomOne) \ou (\posi = \roomTwo)$, 
we have 
$\etat_0,\etat_1\in \valuation(p)$ and 
$\etat_i\in \valuation(p')$, for $3\leq i\leq 7$. 

The path 
$\rho= \etat_0\etat_1\etat_2\etat_3\etat_4\etat_5\etat_6\etat_7$, corresponds to the 
scenario depicted using solid line in Fig.~\ref{subfig:csbuildb}.
The set $\{\etat_0\etat_1\etat_2\etat_3\etat_4,\, 
\etat_3\etat_4\etat_5\etat_6\etat_7\}$ is an example of parallel decomposition of $\rho$. 
One can notice that $\rho$ is an elementary path. 
 To show that while being in Room1 the attacker can turn off 
 but also turn on the 
 camera, we could add the transition $(\etat_4,\etat_3)$ to $\Transition$. 
 In this case, the attacker could also take the path 
 $\rho'= \etat_0\etat_1\etat_2\etat_3\etat_4\etat_3\etat_4\etat_5\etat_6\etat_7$
  which is not elementary because it contains 
 the cycle $\etat_3\etat_4\etat_3$.

\end{example}

\subsection{Attack trees}
\label{sec:at}

To evaluate the security of systems, we use attack trees.  An attack tree does not replace 
the state-transition system model -- it complements it with additional 
information on how the corresponding real-life system could be attacked. 
There 
exist a plethora of methods 
and algorithms for quantitative and qualitative reasoning about security 
using attack trees~\cite{survey}.
However, accurate results can only be obtained if the attack tree 
is in some sense consistent with the analyzed system. Our goal is thus to validate the relevance of an attack tree with respect to a 
given system. To make this validation possible, we need a model capturing 
more information than 
just text-based names of the nodes. In this section, we therefore 
introduce a formal definition of attack trees, 
where the difference with the classical definition is the
presence of a goal of the form $\ag{\atinit}{\atfinal}$ at each node.

\begin{definition}[Attack tree]
  \label{def:aatree}
  An \emph{attack tree} $\atree$ over the set of propositions 
  $\AtmSet$ is either a leaf $\ag{\atinit}{\atfinal}$, where
  $\atinit,\atfinal \in \AtmSet$, or a composed tree of the form
  $(\ag{\atinit}{\atfinal},\OP) (\atree_1, \atree_2, \dots, \atree_n)$,
  where $\atinit,\atfinal \in \AtmSet$, $\OP \in \Set{\OR, \AND, \SAND}$ has arity $n \geq 2$, and $\atree_1$, $\atree_2$, \dots, $\atree_n$ are attack trees.
  The \emph{main goal} of an attack tree
  $\atree=(\ag{\atinit}{\atfinal},\OP) (\atree_1, \atree_2, \dots,
  \atree_n)$ is $\ag{\atinit}{\atfinal}$ and its \emph{operator} is
  $\OP$.

The size of an attack tree $\size{\atree}$ is the number of the nodes in
$\atree$.
Formally, 
$\size{\ag{\atinit}{\atfinal}} = 1$ and
$\size{(\ag{\atinit}{\atfinal},\OP) (\atree_1, \atree_2, \dots
  ,\atree_n)} = 1 + \Sigma_{i=1}^n\size{\atree_i}$.
\end{definition}

As an example, the tree in Fig.~\ref{fig:formal_atree} is 
$\atree =
(\ag{\atinit}{\atfinal},\OR) (\atree_1, \atree_2)$. 
The subtree $\atree_1 =
\ag{\atinit_1}{\atfinal_1}$ is a leaf and $\atree_2=
(\ag{\atinit_2}{\atfinal_2},\AND) (\ag{\atinit_{21}}{\atfinal_{21}},
\atree_{22})$ is a composed tree with \linebreak
$\atree_{22} =
(\ag{\atinit_{22}}{\atfinal_{22}},\SAND)
(\ag{\atinit_{221}}{\atfinal_{221}},
\ag{\atinit_{222}}{\atfinal_{222}},
\ag{\atinit_{223}}{\atfinal_{223}})$.\\

Before introducing properties that address correctness of an attack tree,
we need to define the \emph{path semantics} of goal expressions that
arise from tree descriptions. A \emph{goal expression} is either a mere
atomic goal of the form $\ag{\atinit}{\atfinal}$ or a composed goal of
the form $\OP(\ag{\atinit_1}{\atfinal_1},\ag{\atinit_2}{\atfinal_2},\dots,\ag{\atinit_n}{\atfinal_n})$, where $\OP \in
\Set{\OR,\SAND,\AND}$. The path semantics of a goal expression is defined as follows.
\begin{itemize}
  \item
    $\agsem{\ag{\atinit}{\atfinal}} = \Set{\pathvar \in \SPathSet
      \suchthat \pathvar \text{ goes from } \atinit \text{ to } \atfinal}$
  \item
    $\agsem{\OR(\ag{\atinit_1}{\atfinal_1},\ag{\atinit_2}{\atfinal_2},\dots,\ag{\atinit_n}{\atfinal_n})} = \agsem{\ag{\atinit_1}{\atfinal_1}} \cup
    \agsem{\ag{\atinit_2}{\atfinal_2}} \cup \ldots \cup \agsem{\ag{\atinit_n}{\atfinal_n}}$
  \item
    $\agsem{\SAND(\ag{\atinit_1}{\atfinal_1},\ag{\atinit_2}{\atfinal_2},\dots,\ag{\atinit_n}{\atfinal_n})} = \agsem{\ag{\atinit_1}{\atfinal_1}}
    . \agsem{\ag{\atinit_2}{\atfinal_2}}. \ldots .\agsem{\ag{\atinit_n}{\atfinal_n}}$
  \item
    $\agsem{\AND(\ag{\atinit_1}{\atfinal_1},\ag{\atinit_2}{\atfinal_2},\dots,\ag{\atinit_n}{\atfinal_n})} = \Set{ \pathvar \in
      \SPathSet \suchthat \forall i\in \{1,\dots, n\}\ \exists
       \pathvar_i \in \agsem{\ag{\atinit_i}{\atfinal_i}} $, s.t. 
      $ \Set{\pathvar_1,\pathvar_2,\dots ,\pathvar_n} $ is a parallel
      decomposition of $\pathvar}$.
  \end{itemize}

Consider the goal $\ag{\atinit}{\atfinal}$ of our running example, and
let $\runningsystem$  be the system introduced in Example~\ref{ex:system}.
We have
$\agsem{\ag{\atinit}{\atfinal}} = \{\etat_0\etat_1\etat_2(\etat_3\etat_4)^k
\etat_5\etat_6\etat_7 \mid 
k\geq 1\}$, where $(\etat_3\etat_4)^k$ is the path 
composed of $k$ executions of $\etat_3\etat_4$.

\section{Correctness properties of attack trees}
\label{sec:aat}
We now define four correctness properties 
for attack trees,  
illustrate them on our running example, 
and discuss their relevance for 
real-life security analysis.

\subsection{Definitions}

Before formalizing the correctness properties for attack trees, 
we wish to discard attack trees with ``useless''
nodes.  To achieve this, we define the \emph{admissibility} of an
attack tree $\atree$ w.r.t.\ the system $\system$.\\

The property that an attack tree $\atree$ is \emph{admissible} w.r.t.\ a system $\system$ is inductively defined
as follows. A leaf tree $\ag{\atinit}{\atfinal}$ is \emph{admissible} whenever
$\agsem{\ag{\atinit}{\atfinal}} \neq \emptyset$. A composed tree $(\ag{\atinit}{\atfinal},\OP)
(\atree_1, \ldots, \atree_n)$ is \emph{admissible}
whenever three conditions hold: (a)
$\agsem{\ag{\atinit}{\atfinal}} \neq \emptyset$, (b)
$\agsem{\OP(\ag{\atinit_1}{\atfinal_1},\ldots, \ag{\atinit_n}{\atfinal_n})} \neq \emptyset$,
where $\ag{\atinit_i}{\atfinal_i}$ is the main goal of $\atree_i$ ($1\leq i\leq n$), and (c) every subtree $\atree_i$ is admissible.\\

We now propose four notions of correctness, that 
provide various formal meanings to the local refinement of a 
goal in an admissible tree.
\begin{definition}[Correctness properties]
  \label{def:correctness}~\\ Let $\atree$ be a composed admissible
attack tree of the form $(\ag{\atinit}{\atfinal},\OP)
(\atree_1, \atree_2 \dots, \atree_n)$, and assume
$\ag{\atinit_i}{\atfinal_i}$ is the main goal of $\atree_i$, for
$i\in\{1,\dots,n\}$.  The tree $\atree$ has the
\begin{enumerate}
\item \emph{Meet property} if $\agsem{\OP(\ag{\atinit_1}{\atfinal_1},\dots, \ag{\atinit_n}{\atfinal_n})} \inter \agsem{\ag{\atinit}{\atfinal}} \neq \emptyset$.

\item \emph{Under-Match property} if $\agsem{\OP(\ag{\atinit_1}{\atfinal_1},\dots, \ag{\atinit_n}{\atfinal_n})} \subseteq \agsem{\ag{\atinit}{\atfinal}}$.
  
\item \emph{Over-Match property} if $\agsem{\OP(\ag{\atinit_1}{\atfinal_1},\dots, \ag{\atinit_n}{\atfinal_n})} \supseteq \agsem{\ag{\atinit}{\atfinal}}$.

\item \emph{Match property} if $\agsem{\OP(\ag{\atinit_1}{\atfinal_1},\dots, \ag{\atinit_n}{\atfinal_n})} = \agsem{\ag{\atinit}{\atfinal}}$.
\end{enumerate}
\end{definition}

Clearly the Match property implies all other properties, whereas
Under- and Over-Match properties are incomparable -- as illustrated
in Sect.~\ref{sec:matching_example} -- and they both imply the Meet
property. Note that a tree $\atree$ has the Match property if, and only if, it
has both the Under-Match property and the Over-Match property.

The correctness properties of Definition~\ref{def:correctness} are \emph{local} (at the root of the subtree), but they can easily be made 
\emph{global} by propagating their requirement to all of the subtrees. As there are $\size{T}$ many
subtrees, the complexity of globally deciding these properties has the same order of magnitude as in the local case. 

\subsection{Illustration on the running example}
\label{sec:matching_example}
In the system $\runningsystem$ 
defined in Example~\ref{ex:system} and 
composed of the states $\etat_0,\dots
,\etat_7$, we add two states. First, the state $\etat'_0$ that is
similar to $\etat_0$ except that we assume that the window is
open, \ie $\wind = \true$, and second, the state $\etat'_7$ that is
similar to $\etat_0'$ except that we assume that the attacker is in
Room2, \ie $\posi=\roomTwo$.
As a consequence the transitions of the system $\runningsystem$ become
$\etat'_0\rightarrow \etat_0\rightarrow \etat_1\rightarrow \etat_2\rightarrow 
\etat_3\leftrightarrow \etat_4\rightarrow \etat_5\rightarrow
\etat_6\rightarrow \etat_7$ and $\etat'_0\rightarrow \etat'_7$, where the latter models that 
if the window is open, the attacker can reach Room2 undetected by
entering through the window.

Let us consider the attack tree 
$\atree(\ag{\atinit}{\atfinal}, \OR)(\agelem[1],\atree_2)$ from 
Fig.~\ref{fig:formal_atree}, where the main goal
of $\atree_2$ is $\ag{\atinit_{2}}{\atfinal_{2}}$. Since in system
$\runningsystem$, the set of paths $\agsem{\ag{\atinit}{\atfinal}}$ is
exactly the union of $\agsem{\ag{\atinit_1}{\atfinal_1}}$ and
$\agsem{\ag{\atinit_2}{\atfinal_2}}$, the tree $\atree$ has the Match
property w.r.t.\ $\runningsystem$. This means that in order to achieve
goal $\ag{\atinit}{\atfinal}$, it is necessary and sufficient to
achieve goal $\ag{\atinit_1}{\atfinal_1}$ or goal
$\ag{\atinit_2}{\atfinal_2}$ .

We now consider the sub-tree $\atree_2$ of $\atree$ rooted at the
node labeled by $\ag{\atinit_{2}}{\atfinal_{2}}$ in
Fig.~\ref{fig:formal_atree}. The tree $\atree_2$ is of the form
$(\ag{\atinit_{2}}{\atfinal_{2}},\AND)({\ag{\atinit_{21}}{\atfinal_{21}}},
\atree_2')$ where the main goal of $\atree_2'$ is $\ag{\atinit_{22}}{\atfinal_{22}}$. Our objective is to analyze 
the relationship between the main goal $\ag{\atinit_{2}}{\atfinal_{2}}$ of $\atree_2$
and the composed goal
$\AND(\ag{\atinit_{21}}{\atfinal_{21}},\ag{\atinit_{22}}{\atfinal_{22}})$.
In other words, we ask how does the aim of reaching Room2 undetected via
building relates with turning off the camera
($\ag{\atinit_{21}}{\atfinal_{21}}$) and reaching Room2
($\ag{\atinit_{22}}{\atfinal_{22}}$). A quick analysis of system
$\runningsystem$ shows that indeed achieving both subgoals
$\ag{\atinit_{21}}{\atfinal_{21}}$ and
$\ag{\atinit_{22}}{\atfinal_{22}}$ is necessary to achieve goal
$\ag{\atinit_{2}}{\atfinal_{2}}$, but actually it is not
sufficient. Consider the path
$\delta=\etat'_0\etat_0\etat_1\etat_2\etat_3\etat_4\etat_5\etat_6\etat_7$. This
path achieves goal
$\AND({\ag{\atinit_{21}}{\atfinal_{21}}},{\ag{\atinit_{22}}{\atfinal_{22}}})$,
as it can be decomposed into 
$\delta_{21}=\etat'_0\etat_0\etat_1\etat_2\etat_3\etat_4$ and 
$\delta_{22}=\etat_0\etat_1\etat_2\etat_3\etat_4\etat_5\etat_6\etat_7$,
achieving $\ag{\atinit_{21}}{\atfinal_{21}}$ and
$\ag{\atinit_{22}}{\atfinal_{22}}$, respectively.  However,
$\delta\notin \agsem{\ag{\atinit_{2}}{\atfinal_{2}}}$, since
$\etat'_0\not\in \valuation(\atinit_{2})$ (recall that $\atinit_2$
requires the window to be closed which is not the case in $\etat'_0$).
This is what the Over-Match property 
reflects. As a consequence, the main tree $\atree$ does not have
the global Match property w.r.t.\ $\runningsystem$.

Symmetrically to the Over-Match property, Under-Match reflects a sufficient
but not necessary condition.
To illustrate the Under-Match property, let us extend the system 
	$\runningsystem$
	with one more variable $\roofO$ modeling that there is a window in the 
roof 
	of Room2, which is either open ($\roofO=\true$) or closed ($\roofO=\false$). 
	In all previously considered states, we assume that $\roofO=\false$, and we introduce 
	two additional states 
	\begin{itemize}
	\item[$\etat_8$:] is like $\etat_0$ except that $\roofO=\true$;
	\item[$\etat_9$:] is like $\etat_0$ except that $\roofO=\true$ and $\posi=\roomTwo$.
	\end{itemize}
	The extended system, that we denote by $\runningsystem'$, contains all transitions 
already present in $\runningsystem$, as well as 
$\etat_8 \rightarrow \etat_9$ 
which models 
that it is also possible to reach Room2 
undetected by entering thorough the window located in the roof,  
if it is open. 

Let us modify slightly the attack tree from Fig.~\ref{fig:formal_atree} 
by adding to the preconditions $\atinit_1$ and $\atinit_2$ that 
the window in the roof is closed. Formally, we obtain $\atinit'_1$ 
and $\atinit'_2$ as follows
\begin{align*}
  \atinit'_1 \egdef &(\posi = \outside)\et (\keyOne = \false)\et
                     (\keyTwo = \false)\\\et &(\lockOne = \true)
                                                   \et (\lockTwo = \true)\et (\cameraOn = \true)\\ \et &(\wind = \true) \et (\roofO=\false) \\[5pt]
  \atinit'_2 \egdef &(\posi = \outside)\et (\keyOne = \false)\et
                     (\keyTwo = \false)\\\et &(\lockOne = \true)
                                                   \et (\lockTwo = \true)\et (\cameraOn = \true)\\ \et &(\wind = \false \et (\roofO=\false). 
\end{align*}
We are now interested in validating 
(with respect to the extended system $\runningsystem'$) 
the modified tree 
whose root $\ag{\atinit}{\atfinal}$ is refined into 
$\OR{\ag{\atinit'_{1}}{\atfinal_{1}}}{\ag{\atinit'_{2}}{\atfinal_{2}}}$. 
We observe that 
$\delta'=\etat_8\etat_9$ satisfies the goal $\ag{\atinit}{\atfinal}$ 
of its root, \ie entering Room2 undetected starting from outside of the building, 
but $\delta'$ does not satisfy the semantics of 
the disjunctive refinement
$\OR{\ag{\atinit'_{1}}{\atfinal_{1}}}{\ag{\atinit'_{2}}{\atfinal_{2}}}$
because $\etat_8\not\in\valuation(\atinit'_{1})$ nor $\etat_8\not\in\valuation(
\atinit'_{2})$. 
We can thus conclude that 
	$\runningsystem' \models \OR{\ag{\atinit'_{1}}{\atfinal_{1}}}{\ag{\atinit'_{2}}{\atfinal_{2}}} \Vsubseteq \ag{\atinit}{\atfinal}$ which shows that 
	the modified tree has the Under-Match property in $\runningsystem'$. 

Regarding the Meet property, we
invite the reader to consider the following discussion on the relevance of
the correctness properties we have proposed.

\subsection{Relevance of the correctness properties} 
The main objective of introducing the four correctness properties is
to be able to validate an attack tree with respect to a system
$\system$, \ie verify how faithfully the tree represents potential
threats on $\system$.  This is of special importance for the trees
that are created manually or which are borrowed from an attack tree
library.

In the perfect world, we would expect to work with attack trees having 
the (global) Match property, \ie where 
the refinement of every (sub-)goal
covers perfectly all possible ways of reaching the (sub-)goal in the system. 
However, a tree created by a human will rarely 
have this property. The experts usually do not have 
perfect knowledge about the system and might lack information about 
some relevant data.  
Trees that have been created for similar systems are often reused
but they might actually be incomplete or inaccurate 
with respect to the current system. 
Finally, requiring 
the (global) Match property might 
also be unrealistic for goals expressed only with a couple 
$\ag{\texttt{precondition}}{\texttt{postcondition}}$. 
Therefore, Match is often too strong to be the property expected by default.

In practice, experts base their trees on 
some example scenarios, which implies that they obtain trees having 
the (global) Meet property. The Meet property 
-- which ensures that there is at least one path in the system satisfying 
both the parent goal and its refinement --
is the minimum that we expect from 
an attack tree so that we can consider that 
it is (in some sense) correct and 
so that we can start reasoning about 
the security of the underlying system. 

However, in order to be able to perform a thorough and accurate
analysis of security, one needs stronger properties to hold.  One of
the purposes of attack trees is to provide a summary of possible
individual attack scenarios in order to quantify the security-relevant
parameters, such as their cost, their time or their probability. This
helps the security experts to compare and rank the different
scenarios, to be able to deduce the most probable ones and propose
suitable countermeasures. The classical bottom-up algorithm for
quantification of attack trees, described for instance
in~\cite{mauw2005foundationsatrees}, assigns the parameter values to
the leaf nodes and then propagates them up to the root, using
functions that depend on the type of the refinement used (in our
case \OR, \AND, \SAND). This means that the value of the parent node
depends solely on the values of its children.  To make such a
bottom-up quantification meaningful from the attacker's perspective,
we need to require at least the (global) Under-Match property.
Indeed, this property stipulates that all the paths satisfying a
refinement of a node's goal also satisfy the goal itself.  Under-Match
corresponds thus to an under-approximation of the set of scenarios and
it is enough to consider it for the purpose of finding a vulnerability
in the system.

To make the analysis meaningful from the point 
of view of the defender, we will rather require the Over-Match property. 
This property means that all the paths satisfying the parent goal 
also satisfy its decomposition into sub-goals. 
Since the Over-Match property corresponds to an 
over-approximation of the set of scenarios, it is enough to consider it for 
the purpose of designing countermeasures.

Our method to evaluate the correctness of an attack tree is to check
Admissibility and the (global) Meet property. If it holds, then we say
that the attack tree construction is correct w.r.t.\ to the analyzed
system. We then look at the stronger properties.  Depending on the
situation, the expert might want to ensure either the (global)
Under-Match or the (global) Over-Match property. If the tree fails to
verify the desired property with respect to a given system $\system$,
then it needs to be reshaped before it can be employed for the
security analysis of the real system modeled by $\system$.

\section{Complexity issues}
\label{sec:logiccomplexity}

In this section, we address the complexity of deciding our four
correctness properties introduced in Definition~\ref{def:correctness}.
Table~\ref{tab:complexityrecap} gives an
overview of the obtained results.  In the case of the \OR and
the \SAND operators, all the correctness properties are decided in
polynomial time, which is promising in practice.  However, for
the \AND operator, checking the Admissibility property and the Meet
property is \NP-complete, and checking the Under-Match property is
\coNP-complete. These last two problems are therefore
intractable~\cite{garey2002computers}, but recall that their complexity in
practice might be lower thanks to much favorable kinds of instances
(see for example \cite{leyton2014understanding}).

\begin{table}[h]
  \centering
	 \resizebox{0.8\textwidth}{!}{
  \begin{tabular}{|c| c| c| c| c| c|}
    \cline{2-6}
        \multicolumn{1}{r|}{} & Admissibility & Meet & Under-Match & Over-Match & Match \\
    \hline
    \OR & \PTIME & \PTIME & \PTIME & \PTIME & \PTIME \\
    \hline
    \SAND & \PTIME & \PTIME & \PTIME & \PTIME & \PTIME \\
    \hline
    \AND & \NP-c & \NP-c & \coNP-c & \coNP & \coNP \\
    \hline
  \end{tabular}
	}
  \caption{Complexities of the correctness properties.}
  \label{tab:complexityrecap}
\end{table}

We first state two lemmas that will be useful for our complexity analysis. 
Lemma~\ref{lem:pathsizeflat} provides a bound to the size of paths we need
to consider in the system for the verification of correctness properties. 
Lemma~\ref{lem:SANDANDcheckP} provides 
the complexity of checking if a path reflects a particular combination of subgoals. 

\begin{lemma}
  \label{lem:pathsizeflat}
  Let $\system$ be a transition system, $\OP\in\Set{\OR,\AND,\SAND}$, and $\atinit_1,\atfinal_1,\ldots,
  \atinit_n,\atfinal_n \in \AtmSet$. For every path $\pathvar$ in $\agsem{\OP(\ag{\atinit_1}{\atfinal_1},\ldots,\ag{\atinit_n}{\atfinal_n})}$,
  there exists a path $\pathvar' $ of linear size in $\size{\system}$ and $n$ that is also in $
  \agsem{\OP(\ag{\atinit_1}{\atfinal_1},\ldots,\ag{\atinit_n}{\atfinal_n})}$
  and which preserves the ends of $\pathvar$, i.e.,   $\pathvar'(0)=\pathvar(0)$ and
  $\pathvar'(\size{\pathvar'})=\pathvar(\size{\pathvar})$. More precisely,
  $\size{\pathvar'} \in \OO{(2n - 1)\size{\StateSet}}$.
\end{lemma}

 \begin{proof} Let us start with the following remark
    \begin{remark}
       \label{rem:elem}
       Note that if $\pathvar \in \agsem{\ag{\atinit}{\atfinal}}$, then the path
      $\pathvar'$ obtained by removing all the cycles of $\pathvar$ is also in
      $\agsem{\ag{\atinit}{\atfinal}}$, and that $\pathvar(0) = \pathvar'(0)$ and
      $\pathvar(\size{\pathvar}) = \pathvar'(\size{\pathvar'})$. By construction, we also have 
      $\size{\pathvar'} \leq \size{\StateSet}$.
     \end{remark}
    
     We make a proof for each on the three operators $\OR$, $\SAND$, and $\AND$.
    \begin{itemize}
    \item 
      Let $\pathvar \in
      \agsem{\OR(\ag{\atinit_1}{\atfinal_1},\ldots,
        \ag{\atinit_n}{\atfinal_n})}$. Then, there 
				exists an $i \in \Setleqn{n}$, 
      such that $\pathvar \in \agsem{\ag{\atinit_i}{\atfinal_i}}$.  We
      apply Remark~\ref{rem:elem} to construct a path $\pathvar'$ with
      $\pathvar(0) = \pathvar'(0)$, $\pathvar(\size{\pathvar}) =
      \pathvar'(\size{\pathvar'})$, and $\size{\pathvar'} \leq \size{\StateSet}
      \in \OO{(2n - 1)\size{\StateSet}}$.
      
    \item 
      Let $\pathvar \in
      \agsem{\SAND(\ag{\atinit_1}{\atfinal_1},\ldots,
        \ag{\atinit_n}{\atfinal_n})}$. Then, 
				there exists a concatenation
      $\pathvar =$\linebreak 
			$\pathvar_1 . \pathvar_2 . \ldots . \pathvar_n$, such
      that $\pathvar_i \in \agsem{\ag{\atinit_i}{\atfinal_i}}$, for all $i
      \in \Setleqn{n}$. By Remark~\ref{rem:elem}, we can build $n$
      elementary paths $\pathvar_i'$ each of them being an element of
    $\agsem{\ag{\atinit_i}{\atfinal_i}}$ preserving the extremal
    states of their corresponding $\pathvar_i$. The concatenation
    $\pathvar'=\pathvar_1'\pathvar_2' \ldots \pathvar_n'$ is well
    defined, and belongs to
    $\agsem{\SAND(\ag{\atinit_1}{\atfinal_1},\ldots,
      \ag{\atinit_n}{\atfinal_n})}$. 
    Since for each $i$, $\size{\pathvar_i'} \leq \size{\system}$, we have $\size{\pathvar'} \leq
    n \size{\system} \in \OO{(2n - 1)\size{\StateSet}}$.
    Also note that
    $\pathvar(0)=\pathvar_1(0)=\pathvar_1'(0)=\pathvar'(0)$ and
    $\lst{\pathvar}=\lst{\pathvar_n}=\lst{\pathvar_n'}=\lst{\pathvar'}$.
    
  \item 
    Let $\pathvar \in
    \agsem{\AND(\ag{\atinit_1}{\atfinal_1},\ldots,
      \ag{\atinit_n}{\atfinal_n})}$. Then, there exist
			the anchoring intervals\linebreak
    $\interval{k_1}{l_1},\dots , \interval{k_n}{l_n}$
     of a parallel decomposition of $\pathvar$, such
    that, for each $1 \leq i \leq n$, $\pathvar(k_i) \in
    \valuation(\atinit_i)$ and $\pathvar(l_i) \in
    \valuation(\atfinal_i)$.
    
    We let $m_1 \leq m_2 \leq \ldots \leq
    m_{2n}$ be the result of sorting the elements $k_1,l_1, \dots ,
    k_n,l_n$. Note that necessarily $m_1=0$ and $m_{2n}=\size{\pathvar}$,
    and that $\pathvar$ equals to the concatenation
    $\pathvar\interval{m_1}{m_2}. \pathvar\interval{m_2}{m_3} . \dots
    \pathvar\interval{m_{2n-1}}{m_{2n}}$. For $1 \leq r
    \leq 2n-1$, let $\rho_r$ be the elementary path obtained from
    $\pathvar\interval{m_r}{m_{r+1}}$ by removing all cycles (we may have
    get a single state if $\pathvar\interval{m_r}{m_{r+1}}$ was a cycle,
    so that $\size{\rho_r}=0$).
    
    Clearly, the concatenation
    $\pathvar'=\rho_1\ldots\rho_{2n-1}$ is well-defined and
    $\size{\pathvar'}\leq (2n-1)\size{\system}$, since each $\rho_r$ is
    of the size at most $\size{\system}$. Also, by construction, we have
    $\fst{\pathvar'}=\fst{\rho_1}= \fst{\pathvar}$ and
    $\lst{\pathvar'}=\lst{\rho_{2n-1}}=\lst{\pathvar}$.
    
    It remains to be proven
    that $\pathvar' \in \agsem{\AND(\ag{\atinit_1}{\atfinal_1},\ldots,
      \ag{\atinit_n}{\atfinal_n})}$.
    
    Note that, for all $i\in\Setleqn{n}$, the states $\pathvar(k_i)$ and
    $\pathvar(l_i)$ also occur in $\pathvar'$, at positions we now
    characterize.
    
    Let $m'_1 \leq \ldots \leq m'_{2n}$ be the positions in $\pathvar'$
    defined by $m'_1=0$ and $m'_{r+1} = m'_r + \size{\rho_r}$,
    for $r\in\Setleqn{2n-1}$.  
    Note that $\pathvar'(m'_r)=\pathvar(m_r)$ and that
    $m'_{2n}=\size{\pathvar'}$. Also, if $m_r = m_{r+1}$, then
    $\size{\rho_r}=0$, so that $m'_r=m'_{r+1}$. This implies in
    particular that 
		\begin{equation}
		\label{eq:(*)}
		\tag{*}
		\text{if\ } m'_r < m'_{r+1}, \text{\ then\ } m_r < m_{r+1}.
		\end{equation}
    
    Let $k'_i=m'_r$ (resp. $l'_i = m'_r$) whenever $k_i=m_r$
    (resp.\ $l_i=m_r$). Notice that $k'_i\leq l'_i$, because $k_i \leq
    l_i$.
    
    By construction, each factor
    $\pathvar'_i=\pathvar'[k'_i,l'_i] \in
    \agsem{\ag{\atinit_i}{\atfinal_i}}$.
    
    We finish the proof by showing that the intervals
    $\interval{k'_1}{l'_1}, \ldots, \interval{k'_n}{l'_n}$ are the
    anchorings of a parallel decomposition of $\pathvar'$. This amounts
    to showing that, for any
    $[j,j+1] \subseteq \Setleqn[0]{\size{\pathvar'}}$, there exists an $i$,
    such that $\interval{j}{j+1} \subseteq \interval{k'_i}{l'_i}$.
    
    Let $j \in \Setleqn[0]{\size{\pathvar'}-1}$, we define $r(j)$ as the
    greatest $r \in \Setleqn{2n}$, such that $m'_r \leq j$ (it exists
    since $m'_1=0$). Moreover, we have $r(j) < 2n$ since
    $m'_{r(j)}\leq j < \size{\pathvar'}=m'_{2n}$. Therefore, 
    $r(j)+1 \leq 2n$ so that $m'_{r(j)+1}$ is well-defined and is such
    that we have
    $\interval{j}{j+1} \subseteq \interval{m'_{r(j)}}{m'_{r(j)+1}}$.
    
    We establish that this non-trivial
    $\interval{m'_{r(j)}}{m'_{r(j)+1}}$, which contains
    $\interval{j}{j+1}$, is contained in some $\interval{k'_i}{l'_i}$.

    Clearly since $m'_{r(j)}<m'_{r(j)+1}$, we have
    $m_{r(j)}<m_{r(j)+1}$ (see the condition~\eqref{eq:(*)} above).
    
    By the fact that
    $\pathvar \in \agsem{\AND(\ag{\atinit_1}{\atfinal_1},\ldots,
      \ag{\atinit_n}{\atfinal_n})}$, there exists $i$, st.
    $\interval{m_{r(j)}}{m_{r(j)}+1} \subseteq \interval{k_i}{l_i}$.
    
    Let $t_1$ and $t_2$ be such that $k_i=m_{t_1}$ and
    $l_i=m_{t_2}$. Note that we have $m_{r(j)+1}\leq l_i$: indeed,
    because $m_{r(j)}<l_i=m_{t_2}$, we necessarily have
    $r(j)+1\leq t_2$, which implies $m_{r(j)+1} \leq m_{t_2}=l_i$.
    Now, since
    $k_i=m_{t_1} \leq m_{r(j)} < m_{r(j)+1} \leq m_{t_2} =l_i$, we
    have $t_1 \leq r(j)$ and $r(j)+1 \leq t_2$, so that
    $k'_i=m'_{t_1} \leq m'_{r(j)} <  m'_{r(j)+1} \leq m'_{t_2} = l'_i$
    which concludes. \qed
  \end{itemize}
  
 \end{proof}

\begin{lemma}
  \label{lem:SANDANDcheckP} Let $\system$ be a transition system, 
  $\atinit_1,\atfinal_1,\ldots, \atinit_n,\atfinal_n$ be propositions in $\AtmSet$, and let
  $\pathvar \in \SPathSet$. Determining whether
  $\pathvar \in \agsem{\OP(\ag{\atinit_1}{\atfinal_1}, \ag{\atinit_2}{\atfinal_2},\dots, \ag{\atinit_n}{\atfinal_n})}$
  can be done in time $\OO{\size{\pathvar}+ n}$, if $\OP=\SAND$, and in time
  $\OO{\size{\pathvar}n}$, if $\OP=\AND$.
\end{lemma}

  \begin{proof}
    It is easy to show that
    $\pathvar \in \agsem{\SAND(\ag{\atinit_1}{\atfinal_1},
      \ag{\atinit_2}{\atfinal_2},\dots, \ag{\atinit_n}{\atfinal_n})}$
    \textiff we have $\pathvar(0) \in \valuation(\atinit_1)$,
    $\pathvar(\size{\pathvar}) \in \valuation(\atfinal_n)$, and there exists
    a sequence of positions
    $0 \leq k_1 \leq k_2 \leq \ldots \leq k_{n-1} \leq \size{\pathvar}$, with
    $\pathvar(k_i) \in \valuation(\atfinal_i) \inter
    \valuation(\atinit_{i+1})$.  The algorithm checks these three properties
    with a linear exploration of $\pathvar$. One can easily show that it
    requires at most $n + \size{\pathvar} -1$ steps.

    We show an algorithm to check if
    $\pathvar \in \agsem{\AND(\ag{\atinit_1}{\atfinal_1},
      \ag{\atinit_2}{\atfinal_2},\dots, \ag{\atinit_n}{\atfinal_n})}$.
    The algorithm first verifies in time polynomial in
    $\OO{n\size{\pathvar}}$ that, for each $i\in\Setleqn{n}$, there exist
    $k_i,l_i \in \interval{0}{\size{\pathvar}}$, such that
    $\pathvar(k_i) \in \valuation(\atinit_i)$, and
    $\pathvar(l_i) \in \valuation(\atfinal_i)$, and $k_i \leq l_i$. If not,
    the algorithm rejects.  Otherwise, we have
    $\pathvar\interval{k_i}{l_i} \in \agsem{\ag{\atinit_i}{\atfinal_i}}$,
    and we can tune the algorithm with the same complexity so that
    position $k_i$ (resp.\ $l_i$) along path $\pathvar$ is as small
    (resp.\ big) as possible. In other words, the factors
    $\pathvar\interval{k_i}{l_i}$ of $\pathvar$ are the longest ones among
    those in $\agsem{\ag{\atinit_i}{\atfinal_i}}$.
    Now, we show that
    $\pathvar \in
    \agsem{\AND(\ag{\atinit_1}{\atfinal_1},\ag{\atinit_2}{\atfinal_2},\dots,
      \ag{\atinit_n}{\atfinal_n})}$ \textiff
    $\Set{\pathvar\interval{k_1}{l_1},\ldots,\pathvar\interval{k_n}{l_n}}$
    is a parallel decomposition of $\pathvar$ for the anchorings
    $\interval{k_1}{l_1},\ldots,\interval{k_n}{l_n}$.
      
    Clearly, if
    $\Set{\pathvar\interval{k_1}{l_1},\ldots,\pathvar\interval{k_n}{l_n}}$
    is a parallel decomposition of $\pathvar$, then we have
    $\pathvar \in
    \agsem{\AND(\ag{\atinit_1}{\atfinal_1},\ag{\atinit_2}{\atfinal_2},\dots,
      \ag{\atinit_n}{\atfinal_n})}$.

    Conversely, suppose that
    $\pathvar \in
    \agsem{\AND(\ag{\atinit_1}{\atfinal_1},\ag{\atinit_2}{\atfinal_2},\dots,
      \ag{\atinit_n}{\atfinal_n})}$. Let
    $\Set{\pathvar_1,\ldots,\pathvar_n}$ be a parallel decomposition of
    $\pathvar$, with $\pathvar_i \in \agsem{\ag{\atinit_i}{\atfinal_i}}$, for
    $i\in\Setleqn{n}$, and let $\interval{k'_1}{l'_1}$,
    $\interval{k'_2}{l'_2},\dots, \interval{k'_n}{l'_n}$ be the
    respective anchorings of $\pathvar_1,\ldots,\pathvar_n$ in $\pathvar$. By
    the definition of the $k_i$'s and the $l_i$'s, we necessarily have
    $k_i \leq k'_i \leq l'_i \leq l_i$.  Also, as
    $\Set{\pathvar_1,\ldots,\pathvar_n}$ is a parallel decomposition of
    $\pathvar$, from Definition~\ref{def:decomp}, every
    $\interval{j}{j+1} \subseteq \interval{0}{\size{\pathvar}}$ is
    contained in $\interval{k'_i}{l'_i}$, hence in
    $\interval{k_i}{l_i}$. We conclude that
    $\Set{\pathvar\interval{k_1}{l_1},\ldots,\pathvar\interval{k_n}{l_n}}$
    is a parallel decomposition of $\pathvar$ for the anchorings
    $\interval{k_1}{l_1},\ldots,\interval{k_n}{l_n}$.

    We conclude the proof of Lemma~\ref{lem:SANDANDcheckP} by
    resorting on Lemma~\ref{lem:decomp} to check in time
    $\OO{n\size{\pathvar}}$ if
    $\Set{\pathvar\interval{k_1}{l_1},\ldots,\pathvar\interval{k_n}{l_n}}$
    is a parallel decomposition of $\pathvar$ for the respective
    anchorings $\interval{k_1}{l_1},\ldots,\interval{k_n}{l_n}$.
\qed
  \end{proof}

\subsection{Checking Admissibility (column 1 of Table~\ref{tab:complexityrecap})}
We now investigate the complexity of deciding
the admissibility of an attack tree. We first establish two propositions.

\begin{proposition}
  \label{prop:notemptyagORSAND}
  Given a
        system $\system$ and
        $\atinit_1,\atfinal_1,\ldots,\atinit_n,\atfinal_n\in \AtmSet$,
        deciding $\agsem{\ag{\atinit}{\atfinal}} \neq \emptyset$,
        deciding
        $\agsem{\OR(\ag{\atinit_1}{\atfinal_1},\ldots, \ag{\atinit_n}{\atfinal_n})} \neq \emptyset$,
        and deciding
        $\agsem{\SAND(\ag{\atinit_1}{\atfinal_1},\ldots, \ag{\atinit_n}{\atfinal_n})} \neq \emptyset$
        are decision problems in \PTIME.
\end{proposition}

\begin{proof}\ 

\begin{enumerate}
\item \label{enum:caseatg} Determining if $\agsem{\ag{\atinit}{\atfinal}}$ is not empty amounts to performing a standard reachability analysis in $\system$, which can be done in polynomial time.
\item 
By the path semantics of the $\OR$ operator, $\agsem{\OR(\ag{\atinit_1}{\atfinal_1},\ldots,
  \ag{\atinit_n}{\atfinal_n})} \neq \emptyset$ if and only if there is $i \in \Setleqn{n}$, such
  that $\agsem{\ag{\atinit_j}{\atfinal_j}} \neq \emptyset$, which by
  the case~\ref{enum:caseatg} of this proof, yields a polynomial time
  algorithm.
\item 
  Checking that $\agsem{\SAND(\ag{\atinit_1}{\atfinal_1},\ldots,
    \ag{\atinit_n}{\atfinal_n})} \neq \emptyset$ can be done by a
  forward analysis: for $1 \leq i \leq n$, we define a sequence of
  state sets $\StateSet_i$ by induction over $i$ as follows: we let
  $\StateSet_1 = \valuation(\atinit_1)$. Next, for $2 \leq i < n$,
  $\StateSet_{i+1} = \valuation(\iota_{i+1}) \cap
  \valuation(\atfinal_i) \cap \reach(\StateSet_i)$.
  Clearly,
$\agsem{\SAND(\ag{\atinit_1}{\atfinal_1},\ldots,
    \ag{\atinit_n}{\atfinal_n})} \neq \emptyset$ \textiff $S_n \neq \emptyset$.
  Moreover, computing $\StateSet_n$ takes at most $n
  \size{\StateSet}$ steps, since each $\StateSet_{i+1}$ is computed from $\StateSet_i$ in at most
  $\size{\StateSet}$ steps. \qed
  \end{enumerate}
\end{proof}

On the contrary, it is more complex to deal with the \AND operator.

\begin{proposition}
  \label{prop:notemptyAND}
  Given a system $\system$ and $\atinit_1,\atfinal_1,\ldots,\atinit_n,\atfinal_n\in \AtmSet$, deciding the non-emp\-ti\-ness $\agsem{\AND(\ag{\atinit_1}{\atfinal_1},\ldots,
  \ag{\atinit_n}{\atfinal_n})} \neq \emptyset$ is \NPTIME-complete.
\end{proposition}

\begin{proof}\ 

\noindent
  {\bf \NPTIME-easy:}\\ 
	We can use the algorithm of
  Lemma~\ref{lem:SANDANDcheckP}, with the algorithm guessing a path of
  polynomial size according to Lemma~\ref{lem:pathsizeflat}.

\smallskip
\noindent
  {\bf \NPTIME-hard:}\\ 
We recall that a \emph{set of clauses} $\clauses$ over a set of
  (propositional) variables $\Set{p_1, \dots, p_r}$ is composed of
  elements (the clauses) $\clause\in\clauses$ such that $\clause$ is a
  set of literals, that is either a variable $p_i$ or its negation
  $\neg p_i$. The set $\clauses$ is \emph{satisfiable} if there exists
  a valuation of the variables $p_1, \dots, p_r$ that renders all the
  clauses of $\clauses$ true. The \SAT problem is: given a set of
  clauses $\clauses$, to decide if it is satisfiable. It is
  well-known that \SAT is an \NPTIME-complete
  problem \cite{cook1971complexity}.

  Now, let $\clauses = \Set{\clause_1,\dots, \clause_m}$ be a set of
  clauses over variables $\Set{p_1, \dots, p_r}$ (ordered by their index) that is an input of the \SAT problem.
  Classically, we let $\size{\clauses}$ be the sum of the sizes of all the clauses in $\clauses$, where the size of a clause is the number of its literals.

  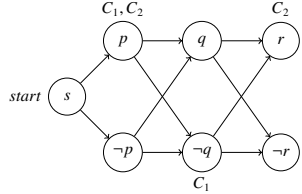
\begin{figure}[ht]
   \begin{center}
 \scalebox{0.7}{
      \begin{tikzpicture}[node distance=1.5cm]
        \node[draw,circle,text width=4mm,align=center] (0) at (0,0)
        {$s$}; \node[draw,circle,text width=4mm,align=center, above
        right of=0] (p1) {$p$}; 
        \node[draw,circle,text width=4mm,align=center, below right
        of=0] (np1) {$\neg p$}; \node[draw,circle,text
        width=4mm,align=center, right of=p1] (p2) {$q$};
        \node[draw,circle,text width=4mm,align=center, right of=np1]
        (np2) {$\neg q$}; \node[draw,circle,text
        width=4mm,align=center, right of=p2] (pl){$r$};
        \node[draw,circle,text width=4mm,align=center, right of=np2]
        (npl) {$\neg r$};

        \draw[->] (0) -> (p1); \draw[->] (0) -> (np1); \draw[->] (p1)
        -> (p2); \draw[->] (p1) -> (np2); \draw[->] (np1) -> (p2);
        \draw[->] (np1) -> (np2); \draw[->] (p2) -> (pl); \draw[->]
        (p2) -> (npl); \draw[->] (np2) -> (pl); \draw[->] (np2) ->
        (npl);

        \node[left of=0, node distance=0.8cm] {$start$}; \node[above
        of=p1, node distance=0.6cm] {$\clause_1,\clause_2$};
        \node[below of=np2, node distance=0.6cm] {$\clause_1$};
        \node[above of=pl, node distance=0.6cm] {$\clause_2$};
\end{tikzpicture}
}
\caption{The system $\system_{\Set{\clause_1,\clause_2}}$ where
  $\clause_1 = p \vee \neg q$ and $\clause_2 = p \vee r$.}
\label{fig:andnphard}
\end{center}
\end{figure}
  In the following, we let the symbol $\literal_i$ denote either $p_i$
  or $\neg p_i$, for every $i \in \{1,\ldots,r\}$.  We define the
  labeled transition system
  $\system_{\clauses}=(\StateSet_\clauses,\Transition_\clauses,\valuation_\clauses)
  $ over the set of propositions $\Set{start,C_1,\dots, C_m}$, where
  $start$ is a fresh proposition, as follows. The set of states is $\StateSet_\clauses
  = \bigcup_{i=1}^r\Set{p_i,\neg p_i} \union \Set{s}$, 
  where $s$ is a fresh state; the transition
  relation is  $\Transition_\clauses
  = \Set{(\literal_i,\literal_{i+1}) \suchthat
  i \in \Setleqn{r-1}} \union \Set{(s,\literal_1)}$; and the labeling of states 
  $\valuation_\clauses: \Set{start,C_1,\dots,
  C_m} \to \parts{\StateSet}$ is such that $\valuation_\clauses(start)=\Set{s}$ and
$\valuation_\clauses(\clause_i) = \Set{\literal \in \clause_i}$ for
  $1 \leq i \leq m$. Note that, by definition, $\size{\system_\clauses}$ is
      polynomial in $\size{\clauses}$.
For example, the transition system corresponding to the set formed by
      clauses $\clause_1= p \vee \neg q$ and $\clause_2= p \vee r$ is
      depicted in Fig.~\ref{fig:andnphard}.

It is then easy to establish that 
$\semantics{\AND(\ag{start}{\clause_1}, \ag{start}{\clause_2}, \dots, \ag{start}{\clause_m})}^{\system_\clauses}\neq \emptyset$ \textiff
$\clauses$ is satisfiable. \qed
\end{proof}
According to the formal definition of the statement
``$\atree$ is admissible w.r.t.\ $\system$'' as defined in
Sect.~\ref{sec:aat}, it is easy to combine the results of
Propositions~\ref{prop:notemptyagORSAND} and \ref{prop:notemptyAND},
to conclude that verifying that a tree is admissible is an \NPTIME-complete problem.

\subsection{Checking the Meet property (column 2 of Table~\ref{tab:complexityrecap})}
\paragraph{Preliminaries on temporal logic.} We consider a syntactic fragment of the temporal logic CTL
   \cite{clarke1981ctl} where the only temporal operator is
   ``eventually'', here denoted by symbol $\ctleventualy$, and where Boolean operators are conjunction and disjunction. The syntax
   of the formulas is $\ctlformula \gramderiv
   p \gramor \ctlformula \ctland \ctlformula \gramor  \ctlformula \ctlor \ctlformula  \gramor \ctleventualy \ctlformula$.
   The semantics of formulas is given
   with regard to a labeled transition system $\system
   = \systemtuple$: each formula $\ctlformula$ denotes a subset of
   states, which we note $\ctlsem{\ctlformula}$, and which is defined
   by induction: $\ctlsem{p} = \valuation(p)$,
   $\ctlsem{\ctlformula \ctland \ctlformula'}
   = \ctlsem{\ctlformula} \cap \ctlsem{\ctlformula'}$,
   $\ctlsem{\ctlformula \ctlor \ctlformula'}
   = \ctlsem{\ctlformula} \cup \ctlsem{\ctlformula'}$, and
   $\ctlsem{\ctleventualy \ctlformula} = \coreach
   ( \ctlsem{\ctlformula} )$, where $\coreach$ is defined in
   Sect.~\ref{sec:TS}. Recall that
   $\state \in \ctlsem{\ctleventualy \ctlformula}$ \textiff, there is a
   path in $\system$ starting from $\state$ and that reaches a state in
   $\ctlsem{\ctlformula}$.
  It is well-established that computing $\ctlsem{\ctlformula}$ can be
  done in polynomial time in $\size{\system}$ and $\size{\ctlformula}$ (see for example \cite{schnoebelen2002complexity}).

We now turn to the complexity of verifying the Meet property.
  
\begin{proposition}
\label{prop:meetORSAND}  
Given a system $\system$ and $\atinit,\atfinal,\atinit_1,\atfinal_1,\ldots,\atinit_n,\atfinal_n\in \AtmSet$, the problem of deciding
$\agsem{\OR(\ag{\atinit_1}{\atfinal_1},\dots, \ag{\atinit_n}{\atfinal_n})} \inter \agsem{\ag{\atinit}{\atfinal}} \neq \emptyset$, and the problem of deciding\\
$\agsem{\SAND(\ag{\atinit_1}{\atfinal_1},\dots, \ag{\atinit_n}{\atfinal_n})} \inter \agsem{\ag{\atinit}{\atfinal}} \neq \emptyset$
are in \PTIME.
\end{proposition}
\begin{proof}\ 

\begin{enumerate}
\item Let $\ctlformula_{\OR} := \ctlbigor_{i=1}^{n} \atinit \ctland
  \atinit_i \ctland \ctleventualy (\atfinal \ctland \atfinal_i)$. We claim that $\agsem{\OR(\ag{\atinit_1}{\atfinal_1},\ldots, 
    \ag{\atinit_n}{\atfinal_n})} \cap \agsem{\ag{\atinit}{\atfinal}}
  \neq \emptyset$ iff $\ctlsem{\ctlformula_{\OR}} \neq \emptyset$.
  We easily conclude our proof from the claim and the fact that
  computing $\ctlsem{\ctlformula_{\OR}}$ can be done in polynomial
  time. 
    We now turn to the proof of the claim.
    Assume that
  $\agsem{\OR(\ag{\atinit_1}{\atfinal_1},\ldots,\ag{\atinit_n}{\atfinal_n})}
  \cap \agsem{\ag{\atinit}{\atfinal}} \neq \emptyset$, and pick some
  path $\pathvar \in
  \agsem{\OR(\ag{\atinit_1}{\atfinal_1},\ldots,\ag{\atinit_n}{\atfinal_n})}
  \cap \agsem{\ag{\atinit}{\atfinal}}$. Then, $\pathvar$ goes from
  $\atinit$ to $\atfinal$ and $\pathvar$ also goes from $\atinit_i$ to
  $\atfinal_i$, for some $i \in \Setleqn{n}$. By the semantics of the
  temporal operators, the state $\pathvar(0)$ is in $\ctlsem{\atinit
    \ctland \atinit_i \ctland \ctleventualy (\atfinal \ctland
    \atfinal_i)}$, entailing $\ctlsem{\ctlformula_{\OR}} \neq
  \emptyset$. Conversely, if $\ctlsem{\ctlformula_{\OR}} \neq
  \emptyset$, then pick some $\state \in
  \ctlsem{\ctlformula_{\OR}}$. By definition of $\ctlformula_{\OR}$,
  there must exist an $i \in \Setleqn{n}$, such that $\state \in
  \ctlsem{\atinit \ctland \atinit_i \ctland \ctleventualy (\atfinal
    \ctland \atfinal_i)}$, entailing the existence of some path from
  $s$, in $\valuation(\atinit)\cap \valuation(\atinit_i)$, to some
  state in $\valuation(\atfinal)\cap \valuation(\atfinal_i)$. Clearly, $\pathvar \in \agsem{\ag{\atinit_i}{\atfinal_i}}
  \cap \agsem{\ag{\atinit}{\atfinal}}$ which is a subset of $\agsem{\OR(\ag{\atinit_1}{\atfinal_1},\ldots,\ag{\atinit_n}{\atfinal_n})}
  \cap \agsem{\ag{\atinit}{\atfinal}}$.
\item 
  Let $\ctlformula_{\SAND} := \atinit \ctland \atinit_1 \ctland
  \ctleventualy (\atfinal_1 \ctland \atinit_2 \ctland \ctleventualy
  (\atfinal_2 \ctland \dots \ctleventualy (\atfinal_n \ctland
  \atfinal)))$. We claim that\\ $\agsem{\SAND(\ag{\atinit_1}{\atfinal_1},\ldots,
    \ag{\atinit_n}{\atfinal_n})} \cap \agsem{\ag{\atinit}{\atfinal}}
  \neq \emptyset$ iff $\ctlsem{\ctlformula_{\SAND}} \neq \emptyset$.
  We easily conclude our proof from the claim and the fact that
  computing $\ctlsem{\ctlformula_{\SAND}}$ can be done in polynomial
  time. 
    It remains to prove the claim.
    Assume $\agsem{\SAND(\ag{\atinit_1}{\atfinal_1},\ldots,
      \ag{\atinit_n}{\atfinal_n})} \cap \agsem{\ag{\atinit}{\atfinal}}
    \neq \emptyset$, and pick some path $\pathvar \in
    \agsem{\SAND(\ag{\atinit_1}{\atfinal_1},\ldots,
      \ag{\atinit_n}{\atfinal_n})} \cap
    \agsem{\ag{\atinit}{\atfinal}}$. Then, $\pathvar$ goes from $\atinit$ to
    $\atfinal$ and it is a concatenation of $n$ paths $\pathvar_i$ going
    from $\atinit_i$ to $\atfinal_i$. Hence, $\pathvar$ successively visits
    $\atinit \ctland \atinit_1$, $\atfinal_1 \ctland \atinit_2$, \dots, and
    $\atfinal_n \ctland \atfinal$. So $\pathvar \in
    \ctlsem{\ctlformula_{\SAND}}$. Conversely, if
    $\ctlsem{\ctlformula_{\SAND}} \neq \emptyset$, then let us 
		pick a state
    $s \in \ctlsem{\ctlformula_{\SAND}}$. By the semantics of the temporal
    operator this shows the existence of a path $\pathvar$ starting from $s$,
    which visits a sequence of $n$ states respectively in $\valuation(\atinit) \cap
    \valuation(\atinit_1)$, $\valuation(\atfinal_1) \cap
    \valuation(\atinit_2)$, \ldots, $\valuation(\atfinal_{n-1}) \cap
    \valuation(\atinit_n)$, and $\valuation(\atfinal_n) \cap
    \valuation(\atfinal)$. Therefore, $\pathvar \in \agsem{\SAND(\ag{\atinit_1}{\atfinal_1},\ldots, 
      \ag{\atinit_n}{\atfinal_n})} \cap \agsem{\ag{\atinit}{\atfinal}}$, which concludes. \qed
\end{enumerate}
\end{proof}

Again, the \AND operator  turns out to be intrinsically more complex to deal with.

\begin{proposition}
  \label{prop:meetAND}
Given a system $\system$ and $\atinit,\atfinal,\atinit_1,\atfinal_1,\ldots,\atinit_n,\atfinal_n\in \AtmSet$, deciding\\
$\agsem{\AND(\ag{\atinit_1}{\atfinal_1},\dots, \ag{\atinit_n}{\atfinal_n})} \inter \agsem{\ag{\atinit}{\atfinal}} \neq \emptyset$ is an \NPTIME-complete problem.
\end{proposition}

\begin{proof} \ 

\smallskip
\noindent
  {\bf \NP-easy:}\\
	We describe a non-deterministic polynomial time algorithm.
  This algorithm guesses a path $\pathvar \in \SPathSet$, of polynomial
  size in $\size{\system}$ and $n$ (this is justified by 
  Lemma~\ref{lem:pathsizeflat}), and checks whether $\pathvar \in
  \agsem{\AND(\ag{\atinit_1}{\atfinal_1},\ldots,\ag{\atinit_n}{\atfinal_n})}$,
  which can be done in polynomial time in the size of $\pathvar$, which is also in polynomial time in
  $\size{\system}$ and $n$ by the choice of $\pathvar$ (see  Lemma~\ref{lem:SANDANDcheckP}).

\smallskip
\noindent
  {\bf \NP-hard:} \\
We now show that deciding $\agsem{\AND(\ag{\atinit_1}{\atfinal_1},\ldots,
  \ag{\atinit_n}{\atfinal_n})} \inter \agsem{\agelem} \neq \emptyset$ is \NP-hard, by reducing the problem of deciding $\agsem{\AND(\ag{\atinit_1}{\atfinal_1},\ldots,
  \ag{\atinit_n}{\atfinal_n})} \neq \emptyset$. Recall that the
  latter is an \NPTIME-hard problem (Proposition~\ref{prop:notemptyAND}). 

Let $\system$ be a transition system, and let $
\atinit_1,\atfinal_1,\ldots,\atinit_n,\atfinal_n\in \AtmSet$. Consider
two fresh propositions, $\atinit$ and $\atfinal$, and define the
system $\system'=(\StateSet,\Transition,\valuation')$ which is like
$\system$ and such that all the states are labeled by both $\atinit$ and $\atfinal$, that is $\valuation'(\atinit)=\valuation'(\atfinal)=\StateSet$. Notice that $\size{\system'}$ is polynomial in $\size{\system}$.

  It is then very easy to establish that
  $\agsem[\system']{\AND(\ag{\atinit_1}{\atfinal_1},\ldots,
    \ag{\atinit_n}{\atfinal_n})} \inter \agsem[\system']{\agelem}$
  \textiff,
$\agsem{\AND(\ag{\atinit_1}{\atfinal_1},\ldots,
  \ag{\atinit_n}{\atfinal_n})} \neq \emptyset$, and this concludes the proof.
\qed
\end{proof}

As a consequence of Propositions~\ref{prop:meetORSAND}
and \ref{prop:meetAND}, it is \NPTIME-complete to verify that an
attack tree has the Meet property, but if we restrict to attack trees
that contain only \OR or \SAND operators, the problem becomes $\PTIME$.

\subsection{Checking the Under-Match property (column 3 of Table~\ref{tab:complexityrecap})} 

The $\OR$ and $\SAND$ operators do not pose any problem.

\begin{proposition}
      \label{prop:subseteqORSAND} Given a system $\system$ and
  $\atinit,\atfinal,\atinit_1,\atfinal_1,\ldots,\atinit_n,\atfinal_n\in \AtmSet$,
  deciding \\
  $\agsem{\OR(\ag{\atinit_1}{\atfinal_1},\dots, \ag{\atinit_n}{\atfinal_n})} \subseteq \agsem{\ag{\atinit}{\atfinal}}$,
  and  deciding  $\agsem{\SAND(\ag{\atinit_1}{\atfinal_1},\dots, \ag{\atinit_n}{\atfinal_n})} \subseteq \agsem{\ag{\atinit}{\atfinal}}$
  are decision problems in \PTIME.
\end{proposition}

  \begin{proof}\ 
	
\noindent
    \textbf{Case for \OR:}\\ Let $\system$ be a transition system and 
  $\atinit,\atfinal,\atinit_1,\atfinal_1,\ldots,\atinit_n,\atfinal_n\in \AtmSet$. 
  By the path semantics of the \OR operator, checking the Under-Match property amounts to deciding 
  \begin{equation}
    \label{eq:eq0}
    \UNION_{i \in \Setleqn{n}}\agsem{\ag{\atinit_i}{\atfinal_i}} \subseteq \agsem{\ag{\atinit}{\atfinal}}.
  \end{equation}
We consider an algorithm that checks that for all $i \in \Setleqn{n}$ the following holds
  \begin{eqnarray}
\label{eq:eq1}    \valuation(\atinit_i) \cap \coreach(\valuation(\atfinal_i)) \subseteq  \valuation(\atinit)\\
\label{eq:eq2}    \valuation(\atfinal_i) \cap \reach(\valuation(\atinit_i)) \subseteq    \valuation(\atfinal).
  \end{eqnarray}
  It accepts the input $\system,
  \atinit,\atfinal,\atinit_1,\atfinal_1,\ldots,\atinit_n,\atfinal_n$ if
  this is the case, and rejects it otherwise.

Notice that the computation of the set $\valuation(\atinit_i) \cap \coreach(\valuation(\atfinal_i))$ and $\valuation(\atfinal_i) \cap
\reach(\valuation(\atinit_i))$ requires only polynomial time in $\size{\system}$ and so does the checking of inclusions in equations~\eqref{eq:eq1} and~\eqref{eq:eq2}.

At this point, the termination of our algorithm is clear.
Regarding its correctness, it is enough to show that the conjunction
of all equations~\eqref{eq:eq1} and~\eqref{eq:eq2} (for all $i$'s) is
equivalent to equation~\eqref{eq:eq0}.  The proof runs easily once one
has noticed that the sets of states $\valuation(\atinit_i) \cap
\coreach(\valuation(\atfinal_i))$ and $\valuation(\atfinal_i) \cap
\reach(\valuation(\atinit_i))$, for $i \in \Setleqn{n}$, 
host the starting and the ending points of all paths in
$\agsem{\ag{\atinit_i}{\atfinal_i}}$, respectively.

\smallskip
\noindent
\textbf{Case for \SAND:}\\
  {\it Claim:} $\agsem{\SAND(\ag{\atinit_1}{\atfinal_1},\ldots,
    \ag{\atinit_n}{\atfinal_n})} \subseteq \agsem{\ag{\atinit}{\atfinal}}$ \textiff
  $\ctlsem{\ctlformula_{\SAND}} = \emptyset$, where
  $$
  \begin{array}{lll}
    \ctlformula_{\SAND} &:= &\ctlnot \atinit \ctland \atinit_1 \ctland \ctleventualy (\atfinal_1 \ctland \atinit_2 \ctland \dots \ctleventualy (\atfinal_n )) \\
    & & \ctlor\\
    & & \atinit_1 \ctland \ctleventualy (\atfinal_1 \ctland \atinit_2 \ctland  \dots \ctleventualy (\atfinal_n \ctland \ctlnot \atfinal)).
  \end{array}
    $$

Assume $\agsem{\SAND(\ag{\atinit_1}{\atfinal_1},\ldots,
  \ag{\atinit_n}{\atfinal_n})} \not\subseteq
\agsem{\ag{\atinit}{\atfinal}}$. This is equivalent to saying that
there exists a path $\pathvar$ which is the concatenation of paths
$\pathvar_1,\ldots,\pathvar_n$ (where $\pathvar_i$ goes from $\atinit_i$ to
$\atfinal_i$), but that does not go from $\atinit$ to $\atfinal$. The
reader can easily deduce that necessarily $\pathvar(0) \in
\ctlsem{\ctlformula_{\SAND}}$, entailing $\ctlsem{\ctlformula_{\SAND}}
= \emptyset$.

  Conversely, let $s \in \ctlsem{\ctlformula_{\SAND}}$. By the
  semantics of $\ctlsem{\ctlformula_{\SAND}}$, one can easily
  verify that there exists a path $\pathvar$ starting from $s$ and which is
  the concatenation of paths $\pathvar_1,\ldots,\pathvar_n$ (where $\pathvar_i$
  goes from $\atinit_i$ to $\atfinal_i$) and such that either $\pathvar$
  does not start from $\atinit$, or $\pathvar$ does not end in  $\atfinal$.

  Whichever latter case we are in, $\pathvar$ is a witness that
  $\agsem{\SAND(\ag{\atinit_1}{\atfinal_1},\ldots, 
    \ag{\atinit_n}{\atfinal_n})}$ is not included in
  $\agsem{\ag{\atinit}{\atfinal}}$, which concludes.
\qed
  \end{proof}

As expected, the  \AND operator yields a more complex problem to solve.
\begin{proposition}
  \label{prop:subseteqAND} Given a system $\system$ and $\atinit,\atfinal,\atinit_1,\atfinal_1,\ldots, \atinit_n,\atfinal_n\in \AtmSet$, deciding\\
$\agsem{\AND(\ag{\atinit_1}{\atfinal_1},\dots, \ag{\atinit_n}{\atfinal_n})} \subseteq \agsem{\ag{\atinit}{\atfinal}}$ is a \coNP-complete problem.
\end{proposition}

  \begin{proof}
	To prove the proposition, we explain how 
		to reduce the problem of deciding if $\agsem{\AND(\ag{\atinit_1}{\atfinal_1},\ldots,\ag{\atinit_n}{\atfinal_n})}$  $\neq \emptyset$
(see Proposition~\ref{prop:notemptyAND}), which is \NP-hard, to the complementary
problem of deciding if $\agsem{\AND(\ag{\atinit_1}{\atfinal_1},\ldots,\ag{\atinit_n}{\atfinal_n})}
\subseteq \agsem{\ag{\atinit}{\atfinal}}$, which is to decide the emptiness of
$\agsem{\AND(\ag{\atinit_1}{\atfinal_1},\ldots,\ag{\atinit_n}{\atfinal_n})}
\setminus \agsem{\ag{\atinit}{\atfinal}}$.

   Let $\system$ be a transition system, and $\atinit_1,\atfinal_1,\ldots,\atinit_n,\atfinal_n\in\AtmSet$. We let $\system'$ be the
   transition system $\system$ with two extra fresh propositions 
   $\atinit$ and $\atfinal$ that label none of the states.

   It can easily be shown that
$\agsem{\AND(\ag{\atinit_1}{\atfinal_1},\ldots,\ag{\atinit_n}{\atfinal_n})}
 = \emptyset$ if, and only if, \linebreak
$\agsem[\system']{\AND(\ag{\atinit_1}{\atfinal_1},\ldots,\ag{\atinit_n}{\atfinal_n})}
\setminus \agsem[\system']{\ag{\atinit}{\atfinal}} = \emptyset$, which trivially holds since 
$\agsem[\system']{\ag{\atinit}{\atfinal}} = \emptyset$ by construction.
\qed
  \end{proof}

\subsection{Checking the Over-Match property (column 4 of Table~\ref{tab:complexityrecap})}
Again, the cases for the \OR and \AND operators are smooth, whereas the
case of the \AND operator  is more difficult.

\begin{proposition}
      \label{prop:supseteqORSAND}
        \label{prop:supseteqAND}
  Given a system $\system$ and $\atinit,\atfinal,\atinit_1,\atfinal_1,\ldots,\atinit_n,\atfinal_n\in \AtmSet$, deciding\\
  $\agsem{\OR(\ag{\atinit_1}{\atfinal_1},\dots, \ag{\atinit_n}{\atfinal_n})} \supseteq \agsem{\ag{\atinit}{\atfinal}}$ and deciding $\agsem{\SAND(\ag{\atinit_1}{\atfinal_1},\dots, \ag{\atinit_n}{\atfinal_n})} \supseteq \agsem{\ag{\atinit}{\atfinal}}$
  are decision problems in \PTIME. On the contrary deciding\\
  $\agsem{\AND(\ag{\atinit_1}{\atfinal_1},\dots, \ag{\atinit_n}{\atfinal_n})} \supseteq \agsem{\ag{\atinit}{\atfinal}}$ is a decision problem in \coNP.
\end{proposition}  

\begin{proof}\ 

\noindent
\textbf{Case for \OR:}\\
Consider the following algorithm which aims at exhibiting a path in
  $\agsem{\ag{\atinit}{\atfinal}}$, going from some state $\state$ to some state
  $\state'$, and such that it is not in 
  $\agsem{\OR(\ag{\atinit_1}{\atfinal_1},\ldots,\ag{\atinit_n}{\atfinal_n})}$.
	
  \begin{algorithmic}[1]
    \ForAll{$\state \in \valuation(\atinit)$}\label{line:loop1}
    \State $R \gets \reach(\state)$ \label{line:post}
    \ForAll{$\state' \in R \inter \valuation(\atfinal)$}\label{line:loop2}
    \State $Witness \gets false$
    \ForAll{$i \in \interval{1}{n}$} \label{line:loop3}
    \State $Witness \gets Witness \ou (\state \in \valuation(\atinit_i) \et \state' \in \valuation(\atfinal_i))$
    \EndFor
    \If{$\non Witness$} \Return{Reject}
    \EndIf
    \EndFor
    \EndFor\\
    \Return{Accept}
  \end{algorithmic}  
	
  This algorithm clearly terminates. Also, it performs a polynomial
  amount of steps in the size of $\system$ and $n$.
	Indeed, line~\ref{line:loop1} is executed at most 
  $\size{\valuation(\atinit)} \leq \size{\StateSet}$ times, the computation of $R$ in line~\ref{line:post} takes polynomial time,
  line~\ref{line:loop2} is executed at most 
  $\size{\valuation(\atfinal)} \leq \size{\StateSet}$ times, and 
	line~\ref{line:loop3} is executed $n$ times. 

  Moreover, one can easily see that the algorithm rejects \textiff
  $Witness$ evaluates to false because of some pair of connected states $\state$
  and $\state'$, entailing the existence of a path in $\agsem{\ag{\atinit}{\atfinal}}$
  but not in any of the $\agsem{\ag{\atinit_i}{\atfinal_i}}$'s.

\smallskip
\noindent
\textbf{Case for \SAND:}\\
  Let
  $\system$ be a transition system, and $\atinit,\atfinal,\atinit_1,\atfinal_1,\ldots,\atinit_n,\atfinal_n\in \AtmSet$.
  We design an algorithm
  based on the verification that
  $\agsem{\ag{\atinit}{\atfinal}} \setminus
  \agsem{\SAND(\ag{\atinit_1}{\atfinal_1},\ldots,\ag{\atinit_n}{\atfinal_n})}
  = \emptyset$.

  Note that, if
  $\pathvar \in \agsem{\ag{\atinit}{\atfinal}} \setminus
  \agsem{\SAND(\ag{\atinit_1}{\atfinal_1},\ldots,\ag{\atinit_n}{\atfinal_n})}$,
  then either 
	\begin{enumerate}
	\item\label{cond:1}
  $\pathvar(0) \in \valuation(\atinit) \setminus \valuation(\atinit_1)$,
  or 
	\item\label{cond:2}
  $\pathvar(\size{\pathvar}) \in \valuation(\atfinal) \setminus
  \valuation(\atfinal_n)$, or 
	\item\label{cond:3}
  $\pathvar(0) \in \valuation(\atinit) \inter \valuation(\atinit_1)$,
  $\pathvar(\size{\pathvar}) \in \valuation(\atfinal)$ but there is an
  $1\leq i\leq n-1$ such that no subsequence of $\pathvar$ visits the
  sets
  $\valuation(\atfinal_1) \inter
  \valuation(\atinit_2),\valuation(\atfinal_2) \inter
  \valuation(\atinit_3),\ldots, \valuation(\atfinal_i) \inter
  \valuation(\atinit_{i+1})$ in this order.
\end{enumerate}

  Notice that the existence of a path that satisfies 
	condition~\ref{cond:1}
  (resp.\ref{cond:2}) is equivalent to $\reach(\valuation(\atinit) \setminus
  \valuation(\atinit_1)) \inter \valuation(\atfinal) \neq \emptyset$
  (resp.\ $\coreach(\valuation(\atfinal) \setminus
  \valuation(\atfinal_n)) \inter \valuation(\atinit) \neq
  \emptyset$). These last two non-emptiness properties can be verified
  in polynomial time in $\size{\system}$, and this is what our algorithm
  first checks.  If one of the two non-emptiness properties holds,
  then the algorithm rejects.  Otherwise the algorithm proceeds by an
  iterated forward search to exhibit a path satisfying 
	condition~\ref{cond:3}, as follows: 
  
  For each $i \in \Setleqn{n-1}$, we denote by $\intermed{i}$
  (``intermediate states for $i$'') the set
  $\valuation(\atfinal_i) \inter \valuation(\atinit_{i+1})$. We now 
  let $(\system_i)_{i\in\Setleqn{n-1}}$ be the transition systems that
  are the restrictions of $\system$ to the set of states outside
  $\intermed{i}$, formally to the set $\StateSet \setminus
  \intermed{i}$.

    We also define the sets of states $G_i$, for $i \in
    \Setleqn[0]{n-1}$, by $G_0 = \valuation(\atinit) \inter
    \valuation(\atinit_1)$, and for $i \in \Setleqn[0]{n-2}$,
    $$G_{i+1}= (G_{i} \union \post(\reach[\system_{i+1}](G_{i}
    \setminus \intermed{i+1}))) \inter \intermed{i+1}.$$

    Intuitively, set $G_{i+1}$ is the set of states in $\system$
    that can be reached by paths visiting a state in $G_i$, thus in
    $Int_i$, that end in $Int_{i+1}$ and such that after their last
    state in $G_i$, theses paths visit 
    $Int_{i+1}$ exactly once. This
    property is formally established by the following lemma.
    
 \begin{lemma}
  \label{lem:good} Let $i \in \Setleqn[0]{n-1}$, and $\state\in\StateSet$.
If $\state \in G_i$, then there exsists a path $\pathvar$
that goes from $G_0$ to $\state$, and, if $i \geq 1$ and $\size{\pathvar} \geq 1$,
$\pathvar\interval{0}{\size{\pathvar}-1}$ does not visit $\intermed{1},\ldots,\intermed{i}$ in this order.
\end{lemma}
\begin{proof} We proceed by induction on $i$. 
The case for $i=0$ is trivial, since $\pathvar=\state$. Let $i=1$ and $\state
\in G_1$. Either $\state \in G_0$, then again $\pathvar=\state$ and we
are done, or $\state \not\in G_0$ and because $G_1=(G_0 \union
\post(\reach[\system_{1}](G_0 \setminus \intermed{1}))) \inter
\intermed{1}$, there exists a path $\pathvar$ of size at least $1$, from
$G_{0}$ to $\state$, such that $\pathvar\interval{0}{\size{\pathvar}-1}$ does
not visit $\intermed{1}$, and we are done.

Let $i>1$ and $\state \in G_i$. If $\state \in G_{i-1}$, then we
conclude by induction hypothesis. Otherwise $\state \not\in G_{i-1}$,
therefore $\state \in \post(\reach[\system_{i}](G_{i-1})$. Thus,
there exists a path $\pathvar''$ going from a state $\state' \in G_{i-1}$ to
$\state$, such that $\size{\pathvar''}\leq 1$, since $\state' \neq
\state$, and $\pathvar''\interval{0}{\size{\pathvar''}-1}$ does not
visit $\intermed{i}$.

Moreover, by induction hypothesis, there exists a path $\pathvar'$ that goes from
$G_0$ to $\state' \in G_{i-1}$, and if $\size{\pathvar'} \geq 1$,
$\pathvar'\interval{0}{\size{\pathvar'}-1}$ does not visit
$\intermed{1},\ldots,\intermed{i-1}$ in this order.  Let $\pathvar = \pathvar' . \pathvar''$,
which goes from $G_0$ to $\state$. We conclude the proof by showing
that $\pathvar\interval{0}{\size{\pathvar}-1}$ does visit
$\intermed{1},\ldots,\intermed{i}$ in this order.

If $\size{\pathvar'} \geq 1$, we know that
$\pathvar'\interval{0}{\size{\pathvar'}-1}$ does not visit
$\intermed{1},\ldots,\intermed{i-1}$ in this order, which implies that
$\pathvar'\interval{0}{\size{\pathvar'}-1}$ does not visit
$\intermed{1},\ldots,\intermed{i}$ in this order. Then, a visit in $\pathvar$ of sets
$\intermed{1},\ldots,\intermed{i}$ in this order would  require to visit
$\intermed{i}$ in $\pathvar''$ which cannot be the case.
  \qed
\end{proof}

Now, for every  $i \in
\Setleqn{n-1}$, we let
$$B_{i} = \reach[\system_i](G_{i-1} \setminus
    \intermed{i}) \inter \valuation(\gamma).$$

    Intuitively, a state $\state \in B_i$ is such that there is path
    from $G_0$ to $\state$, hence ending in $\valuation(\atfinal)$,
    that visits all the $G_j$'s, for $j<i$, in this order.

\begin{lemma}
  \label{lem:bad}
There exists $i \in \Setleqn{n-1}$ such that $B_i\neq\emptyset$ iff 
  there is a path that satisfies condition (3).
\end{lemma}  
\begin{proof}\ 

\noindent
$(\impl)$\\ 
Let $i \in \Setleqn{n-1}$, such that $B_i\neq\emptyset$.  By
  definition of $B_i$, there exists a path $\pathvar''$ going from a state
  $\state' \in G_{i-1}$ to $\state \in \valuation(\atfinal)$ that does not visit
  $\intermed{i}$.  Also, by Lemma~\ref{lem:good}, we have a path $\pathvar'$
  that goes from $G_0$ to $\state' \in G_{i-1}$, such that, if $i \geq 1$
  and $\size{\pathvar'} \geq 1$, $\pathvar'\interval{0}{\size{\pathvar'}-1}$ does
  not visit $\intermed{1},\ldots,\intermed{i-1}$ in this order.  So it is now
  easy to see that any sub-sequence of states of $\pathvar' . \pathvar''$ is
  either a sub-sequence of states of $\pathvar'$, or has its last state
  in $\pathvar''$, so that $\pathvar' . \pathvar''$ satisfies 
	condition~\ref{cond:3}, for $i$.

\noindent
$(\Leftarrow)$\\ 
Let $\pathvar$ be a path that satisfies condition~\ref{cond:3}.  
First, notice that if $\pathvar$ does not visit $\intermed{1}$,
  then $B_1 \neq \emptyset$.  Otherwise, there exists an $i$, such that
  $\pathvar$ visits $\intermed{1},\ldots,\intermed{i}$ in this order.
	Let
  $i'$ be the maximum of such $i$. We show that $B_{i'+1}\neq\emptyset$.

  For $i\in\Setleqn{i'}$, we define $j(i) \in
  \Setleqn[0]{\size{\pathvar}}$ by induction over $i$: $j(1)$ is the
  smallest $j$ such that $\pathvar(j) \in \intermed{1}$, and for
  $i\in\Setleqn[2]{i'}$, $j(i)$ is the smallest $j$ above $j(i-1)$
  such that $\pathvar(j) \in \intermed{i}$. One can verify by
    induction on $i\in\Setleqn{i'}$ that $j(i)$ is well defined and
    that $\pathvar(j(i)) \in G_i$.

  We conclude the proof of the lemma by noticing that by definition of
  $i'$, we have for all $i \geq j(i')$, $\pathvar(i)
  \not\in\intermed{i'+1}$, so that $\lst{\pathvar} \in B_{i'+1}$. Hence
  $B_{i'+1}\neq\emptyset$. \qed
\end{proof}

The algorithm iteratively computes in polynomial time the sets
$G_i$ and $B_i$ and rejects if one of the encountered $B_i$ is empty.
Otherwise the algorithm accepts.

\smallskip
\noindent
\textbf{Case for \AND:}\\
  We exhibit a non-deterministic polynomial time algorithm to decide
  whether or not $\AND(\ag{\atinit_1}{\atfinal_1},\ldots,\ag{\atinit_n}{\atfinal_n})
  \not\supseteq \ag{\atinit}{\atfinal}$.

  The algorithm guesses a path
  $\pathvar \in \agsem{\ag{\atinit}{\atfinal}} \setminus
  \agsem{\AND(\ag{\atinit_1}{\atfinal_1},\ldots,\ag{\atinit_n}{\atfinal_n})}$.
  Such a path can be chosen cycle-free, since removing all its cycles
  still makes the resulting path an element of
  $\agsem{\ag{\atinit}{\atfinal}}$, but cannot make it an element of
  $\agsem{\AND(\ag{\atinit_1}{\atfinal_1},\ldots,\ag{\atinit_n}{\atfinal_n})}$.
  Then the algorithm checks that $\pathvar(0) \in \valuation(\atinit)$
  and $\pathvar(\size{\pathvar}) \in \valuation(\atfinal)$. Finally, the
  algorithm uses the polynomial algorithm of Lemma~\ref{lem:SANDANDcheckP}
  to answer whether 
  $\pathvar \not\in \agsem{\AND(\ag{\atinit_1}{\atfinal_1},\ldots,\ag{\atinit_n}{\atfinal_n})}.$
\qed
\end{proof}

Finally, we can get an upper bound for column 5 of Table~\ref{tab:complexityrecap} 
(the Match property)
by taking the maximum between upper
bound complexities for Under-Match and Over-Match, which achieves the filling of Table~\ref{tab:complexityrecap}.

\section{Conclusion  and future work}
\label{sec:conclusion}
In this work, we have developed and studied a formal setting to assist
experts in the design of attack trees when a particular system is
considered. The system is 
described by a finite transition system 
that reflects its dynamics and whose finite
paths (sequences of states) denote attack scenarios. 
The attack tree nodes are labeled with pairs 
$\ag{\atinit}{\atfinal}$ expressing the attacker's goals in 
terms of pre and postconditions. 
The semantics of attack trees is based on 
sets of finite paths in the transition system. Such sets
of paths can be characterized as a mere reachability condition of the form
``all paths from condition $\atinit$ to condition $\atfinal$'', or by
a combination of those by means of  \OR, \AND, and \SAND operators.

We have exhibited  the Admissibility
property which allows us to check whether it makes sense to analyze  
a given attack tree in the context of a considered system. We then propose 
four natural correctness properties on top of Admissibility, namely
\begin{itemize} 
\item Meet -- the node's refinement 
makes sense in a given system;
\item  Under (resp. Over) Match -- 
 the node's refinement under-approximates 
(resp. over-approximates) the goal of the node in a given system; and 
\item Match -- the node's 
refinement expresses exactly the node's goal in a given system.  
\end{itemize}

While analyzing an attack tree with respect to a system, 
we propose to start by checking whether each of its subtrees  
satisfies the Meet property 
-- this is the minimum that we require from a correct attack tree. If this is 
the case, we can then check how well  the tree refines the main attacker's 
goal, using  (Under- and  Over-) Matching. 
Our study reveals that the highest complexity in such  
analysis is due to conjunctive
refinements (\ie the \AND operator), as opposed to disjunctive and 
sequential refinements, cf. Table~\ref{tab:complexityrecap}. 
The reason is that the semantics that we use in our framework relies on paths in a transition 
system and thus modeling and verification for paths' concatenation  
(used to formalize the \SAND refinements) is much simpler than 
those for parallel 
decomposition (used to formalize the \AND refinements).
Indeed,  the latter requires to analyze the combinatorics of paths 
representing children of a conjunctively refined node.

The framework presented in this paper offers numerous possibilities for 
practical applications in industrial setting.
First, it can be used to estimate the quality of a   
refinement of an attack goal, that an expert could borrow from an attack 
pattern library. The correctness properties introduced in this work 
allow us to evaluate the relevance of often generic refinements in 
the context of a given system.
Second, classical attack trees use 
text-based nodes that represent a desired configuration to be reached
(our postcondition $\atfinal$) without specifying  
the initial configuration (our precondition $\atinit$) 
where the attack will start from. Given a 
transition system $\system$ describing a real system to be analyzed, 
the text-based goals can be straightforwardly translated into 
formal propositions expressing the final configurations (\ie $\atfinal$) to 
be reached by the attacker. The expert may also specify the initial 
configurations (\ie $\atinit$), but if he does not do so, they can be 
automatically generated from the transition system, by simply taking 
all states belonging to the set $\coreach(\valuation(\atfinal))$ of 
predecessors of $\valuation(\atfinal)$ in $\system$.

For pedagogical reasons, we have focused on simple atomic goals 
(\ie node labels) that are definable in terms of a precondition
and a postcondition. As one of the future directions, we
would like to enrich the language of atomic goals, for
instance by adding variables with history or invariants.  
Variables with history can be used to express properties such as
\emph{"Once detected, the attacker will always stay detected"}. 
With invariants, we may add constraints to the goals, as in
\emph{"Reach Room2 undetected without ever crossing
  Room1"}.  If invariants are added to atomic goals, for instance
using LTL formulas, the complexity of some problems presented in this
paper may increase. In that case, checking that a path 
satisfies  the semantics of a node might no longer be done in
constant time, 
  	but in
  polynomial time, or even in PSPACE-complete, if arbitrary LTL formulas are
  allowed \cite{GiacomoVardi2013ltlfinitetraces}.    It would then be relevant
  to study  the interplay between the expressiveness\linebreak of the atomic goals
  and the complexity of verifying these correctness properties.

It would also be interesting to extend our framework to capture more
complex properties than those defined in
Definition~\ref{def:correctness}.  Pragmatic examples of such
properties would be validities and tests expressed in an adequate
logic.
\emph{Validities} would be
formulas that are true in any system. An example of a validity would look like 
$\AND{\ag{\atinit}{\atfinal}}{\ag{\atinit'}{\atfinal'}}
 \Vsupseteq \SAND{\ag{\atinit}{\atfinal}}{\ag{\atinit'}{\atfinal'}}$, with the meaning that a sequential composition is a particular case of 
parallel composition.  
\emph{Tests} would be formulas which might be true in some systems, 
but not necessarily in all cases. For instance, a formula like
$\AND{\ag{\atinit}{\atfinal}}{\ag{
\atinit'}{\atfinal'}} \Vsubseteq \SAND{\ag{\atinit}{\atfinal}}{\ag{\atinit'}{
\atfinal'}}$ would mean that, in a given system, 
it is impossible to realize both $\ag{\atinit}{
\atfinal}$ and $\ag{\atinit'}{\atfinal'}$ otherwise than sequentially in this 
particular order. 

Finally, we are currently working 
on integrating the framework developed in this work 
to the ATSyRA tool. The ultimate goal is to design software for generation of 
attack trees satisfying the correctness properties that we have introduced. The short-
term objective is to validate the practicality of the proposed framework 
and its usability with respect to the complexity results that we have proven 
in this work.

\bibliographystyle{splncs} 
\bibliography{ms}

\begin{thebibliography}{10}
\providecommand{\url}[1]{\texttt{#1}}
\providecommand{\urlprefix}{URL }

\bibitem{post15}
Aslanyan, Z., Nielson, F.: Pareto efficient solutions of attack-defence trees.
  In: {POST}. LNCS, vol. 9036, pp. 95--114. Springer (2015)

\bibitem{aslanyan2017exactcosts}
Aslanyan, Z., Nielson, F.: Model checking exact cost for attack scenarios. In:
  International Conference on Principles of Security and Trust. Springer (2017)

\bibitem{audinot2016soundnessatree}
Audinot, M., Pinchinat, S.: {On the Soundness of Attack Trees}. In: Graphical
  Models for Security. {LNCS}, vol. 9987, pp. 25--38. Springer (2016)

\bibitem{esorics17}
Audinot, M., Pinchinat, S., Kordy, B.: Is my attack tree correct? In:
  {ESORICS}. LNCS, Springer (2017), (to appear)

\bibitem{clarke1981ctl}
Clarke, E.M., Emerson, E.A.: Design and synthesis of synchronization skeletons
  using branching time temporal logic. In: Workshop on Logic of Programs. pp.
  52--71. Springer (1981)

\bibitem{cook1971complexity}
Cook, S.A.: The complexity of theorem-proving procedures. In: Proceedings of
  the third annual ACM symposium on Theory of computing. pp. 151--158. ACM
  (1971)

\bibitem{GiacomoVardi2013ltlfinitetraces}
De~Giacomo, G., Vardi, M.Y.: Linear temporal logic and linear dynamic logic on
  finite traces. In: IJCAI'13 Proceedings of the Twenty-Third international
  joint conference on Artificial Intelligence. pp. 854--860. Association for
  Computing Machinery (2013)

\bibitem{Gadyatskaya2016modelcheckingquantitative}
Gadyatskaya, O., Hansen, R.R., Larsen, K.G., Legay, A., Olesen, M.C., Poulsen,
  D.B.: Modelling attack--defense trees using timed automata. In: {FORMATS}.
  LNCS, vol. 9884, pp. 35--50. Springer (2016)

\bibitem{garey2002computers}
Garey, M.R., Johnson, D.S.: Computers and intractability, vol.~29. W. H.
  Freeman and Company (2002)

\bibitem{FundInf17}
Horne, R., Mauw, S., Tiu, A.: Semantics for specialising attack trees based on
  linear logic. Fundam. Inform.  153(1-2),  57--86 (2017)

\bibitem{trespass_generation}
Ivanova, M.G., Probst, C.W., Hansen, R.R., Kamm{\"{u}}ller, F.: {Transforming
  Graphical System Models to Graphical Attack Models}. In: Graphical Models for
  Security. {LNCS}, vol. 9390, pp. 82--96. Springer (2015)

\bibitem{jhawar2015seqatree}
Jhawar, R., Kordy, B., Mauw, S., Radomirovi{\'c}, S., Trujillo{-}Rasua, R.:
  {Attack Trees with Sequential Conjunction}. In: {SEC}. {IFIP AICT}, vol. 455,
  pp. 339--353. Springer (2015)

\bibitem{JurgensonW09}
J{\"{u}}rgenson, A., Willemson, J.: {Serial Model for Attack Tree
  Computations}. In: {ICISC}. LNCS, vol. 5984, pp. 118--128. Springer (2009)

\bibitem{kordy2012adtrees}
Kordy, B., Mauw, S., Radomirovic, S., Schweitzer, P.: Attack--defense trees. J.
  Log. Comput.  24(1),  55--87 (2014)

\bibitem{survey}
Kordy, B., Pi{\`{e}}tre{-}Cambac{\'{e}}d{\`{e}}s, L., Schweitzer, P.: Dag-based
  attack and defense modeling: Don't miss the forest for the attack trees.
  Computer Science Review  13-14,  1--38 (2014)

\bibitem{bayesian}
Kordy, B., Pouly, M., Schweitzer, P.: Probabilistic reasoning with graphical
  security models. Inf. Sci.  342,  111--131 (2016)

\bibitem{Kumar2015uppaal}
Kumar, R., Ruijters, E., Stoelinga, M.: Quantitative attack tree analysis via
  priced timed automata. In: {FORMATS}. LNCS, vol. 9268, pp. 156--171. Springer
  (2015)

\bibitem{leyton2014understanding}
Leyton-Brown, K., Hoos, H.H., Hutter, F., Xu, L.: Understanding the empirical
  hardness of {NP}-complete problems. Communications of the ACM  57(5),
  98--107 (2014)

\bibitem{mauw2005foundationsatrees}
Mauw, S., Oostdijk, M.: {Foundations of Attack Trees}. In: {ICISC}. LNCS, vol.
  3935, pp. 186--198. Springer (2005)

\bibitem{OWASP}
{OWASP}: {CISO AppSec Guide: Criteria} for managing application security risks
  (2013)

\bibitem{PhillipsS98}
Phillips, C.A., Swiler, L.P.: A graph-based system for network-vulnerability
  analysis. In: {Workshop on New Security Paradigms}. pp. 71--79. {ACM} (1998)

\bibitem{Wolter}
Pieters, W., Padget, J., Dechesne, F., Dignum, V., Aldewereld, H.:
  Effectiveness of qualitative and quantitative security obligations. J. Inf.
  Sec. Appl.  22,  3--16 (2015)

\bibitem{pinchinat2015atsyra}
Pinchinat, S., Acher, M., Vojtisek, D.: {ATSyRa: An Integrated Environment for
  Synthesizing Attack Trees -- (Tool Paper)}. In: Graphical Models for
  Security. LNCS, vol. 9390, pp. 97--101. Springer (2015)

\bibitem{OTAN}
Research, N., (RTO), T.O.: {Improving Common Security Risk Analysis}. Tech.
  Rep. AC/323(ISP-049)TP/193, North Atlantic Treaty Organisation, University of
  California, Berkeley (2008)

\bibitem{act}
Roy, A., Kim, D.S., Trivedi, K.S.: Attack countermeasure trees {(ACT):} towards
  unifying the constructs of attack and defense trees. Security and
  Communication Networks  5(8),  929--943 (2012)

\bibitem{Schn}
Schneier, B.: {Attack Trees: Modeling Security Threats}. Dr. Dobb's Journal of
  Software Tools  24(12),  21--29 (1999)

\bibitem{schnoebelen2002complexity}
Schnoebelen, P.: The complexity of temporal logic model checking. Advances in
  modal logic  4(393-436), ~35 (2002)

\bibitem{SheynerHJLW02}
Sheyner, O., Haines, J.W., Jha, S., Lippmann, R., Wing, J.M.: {Automated
  Generation and Analysis of Attack Graphs}. In: {IEEE} S\&P. pp. 273--284.
  {IEEE} Computer Society (2002)

\bibitem{gal}
Thierry{-}Mieg, Y.: Symbolic model-checking using its-tools. In: {TACAS}. LNCS,
  vol. 9035, pp. 231--237. Springer (2015)

\bibitem{vigo_generation}
Vigo, R., Nielson, F., Nielson, H.R.: {Automated Generation of Attack Trees}.
  In: {CSF}. pp. 337--350. {IEEE} Computer Society (2014)

\end{thebibliography}

\end{document}